\documentclass[11pt]{article}
\usepackage{amssymb,amsmath,latexsym,proof,psfig}
\usepackage{algpseudocode}
\usepackage{makeidx}
\usepackage{MnSymbol,wasysym}
\usepackage{prooftrees}
\usepackage{eqparbox,array}
\usepackage{graphicx}
\usepackage{pstricks}
\usepackage{pst-node}
\usepackage{textcomp}
\usepackage{amsthm} 
\SpecialCoor

\title{A Polynomial Decision for \thsat}
\author{M Angela Weiss\\
\small{weiss@ime.usp.br, 
       Universidade de S\~{a}o Paulo, S\~{a}o Paulo - Brasil}
       }
\date{}
\newtheorem{theorem}{Theorem}[section]
\newtheorem{lemma}[theorem]{Lemma}

\newtheorem{proposition}[theorem]{Proposition}
\newtheorem{definition}[theorem]{Definition}

\newtheorem{example}[theorem]{Example}
\newtheorem{notation}[theorem]{Notation}
 
\newtheorem{observation}[theorem]{Observation}
\newtheorem{procedure}[theorem]{Procedure}

\newtheorem{pseudocode}[theorem]{Pseudocode}

\newcommand{\letter}{{\mathsf Letter}}

\newcommand{\clndr}{{\mathbb Clndr}}

\newcommand{\clsddg}{\mathcal{CSD}}

\newcommand{\thsat}{{\bf 3-sat}}
\newcommand{\twsat}{{\bf 2-sat}}
\newcommand{\sat}{{\bf sat}}

\makeindex
\begin{document}
\maketitle
\begin{abstract}
  We propose a polynomially bounded, in time and space, method
  to decide whether a given \thsat\ formula is satisfiable or not.

  The tools we use here are, in fact, very simple. We first decide
  satisfiability for a particular \thsat\ formula, called {\em pivoted \thsat}
  and, after a plain transformation, still keeping the polynomial
  boundaries, it is shown that \thsat\ formula is unsatisfiable if and only if
  its corresponding pivoted formula is unsatisfiable.
\end{abstract}
\section{Introduction}\label{INTR:sec}

After the preliminary definitions that establish a common language, we  define
a {\em Pivoted \thsat\ formula}\index{Pivoted \thsat\ formula}, the main subject we deal with.
After that, we show that any \thsat\ formula has a stronger version in the format of a Pivoted \thsat in the sense that a Pivoted \thsat\ formula is unsatisfiable if and only if the original \thsat\ formula is unsatisfiable and, if the original \thsat\ formula is a theorem (true under any valuation), its pivoted version is not a theorem

Secondly, we show that we can solve, that is, we can decide if the pivoted formula is {\bf Satisfiable} or {\bf Unsatisfiable}, in a polynomial use of space and time.
Our resolution methods will require no more than basic knowledge of Boolean
Logic and some basic skills.

An excellent way to start on the Complexity topic is by reading Scott Aaronson
breathtaking survey,

https://www.scottaaronson.com/papers/pnp.pdf

The question is, ``What is computing?''. What kind of reasoning, operation,
can be performed using an algorithm'' and how fast the operation can be
performed. See \cite{ASBB} for a general overview, or consult the roots, 
\cite{CK71}, \cite{RMO1960}, \cite{RMO1959}, \cite{HS1965},
\cite{CKRCK79}, \cite{LLII}, \cite{LLI}, \cite{MST72}.

A precise definition of {\em algorithm} was given by Alan Turing in
1937 (see \cite{AT1937}). The natural question that arises is: {\em What is the
  computational difficulty of performing such an algorithm?} See, in
chronological order, \cite{RMO1959}, \cite{RMO1960}, \cite{LLII},
\cite{HS1965}, \cite{ASBB}.
Classification of complexity is found in
\cite{M87}, \cite{PAPA94} and \cite{ASBB}.
We aim to open a new trail to understand
complexity questions on polynomially solving our problem that 
obviously can be solved in exponential time and space and draw new
lines on the complexity classification.

In Section \ref{BCKGRND:sec}, we provide classical definitions of
valuation for Boolean Logic, define SAT, pose some more standard
definitions, and introduce Pivoted \thsat\, the special kind of \thsat\ we focus on. To solve the central problem, decide satisfiability or unsatisfiability, of Pivoted \thsat\, we write a digraph, the {\em cylindrical digraph}, a digraph where disjunctions of formulas play the rule of edges
endowed with labels.
We show that, given a \thsat\ formula $\Psi$, there is a pivoted \thsat\ formula
$\Psi_{T}$ so that $\Psi$ is unsatisfiable if and only if $\Psi_{T}$ is unsatisfiable.

After building the cylindrical digraph, we detach parts of the cylindrical digraph, the {\em closed digraphs} that mark unsatisfiable combinations of formulas in a pivoted \thsat.

In Section \ref{LASTIDEA:sec}, we fully describe the algorithms we use in
the process of deciding whether the closed digraph represents all
$\mathbf{2}^{\mathbf{m}}$ combinations of valuation over the conjugated literals of a Pivoted \thsat.
We will not write down all possible combinations to have a polynomial bound in our solution. Rather,
we analyze the question there is a  satisfiable combination by generating a digraph.

Section \ref{THPVTH} is the connecting link between a \thsat\ formula and a pivoted \thsat\ formula. We demonstrate that given a \thsat\ formula $\Psi$, then there exists a strong version of $\Psi$, $\Psi_{T}$, so that $\Psi$ is unsatisfiable if and only if $\Psi_{T}$ is unsatisfiable.

Section  \ref{BC:sec} is quite a technical section and, as usual, a
profound and critical reading follows after unraveling the basic Section,
\ref{BCKGRND:sec} and the foundation part, Section \ref{LASTIDEA:sec}.

Finally, a very long part of this paper is devoted to the examples, Section
\ref{EX:sec}. This section can be regarded as a companion section.
\section{Basic Definitions}\label{BCKGRND:sec}
We work with a Boolean Language whose basic symbols are
$\vee,\wedge,\neg$ endowed with a finite set of
atoms $\mathcal{A}$. The set of literals, $\mathcal{L}$, is the set
$\mathcal{A}\cup\{\neg p|p\in\mathcal{A}\}$. A pair formed by a literal
together with its negation is called a conjugated pair. Besides a finite set of atoms, we have the symbols $\top$ and $\bot$.

Here we give the (classical) valuation definition over a Boolean
formula. That is the definition we will use henceforth.

\begin{definition}\label{BOOLFOR:def}
Given a finite set of {\em atoms}, $\mathcal{A}$, $\psi$ belongs to
the set of Boolean formulas $\mathtt{F}$ if
\begin{enumerate}
\item $\psi\in\mathcal{A}$;
\item $\psi=\chi\wedge\phi$ and both $\chi$ and $\phi$ belong to $\mathtt{F}$;
\item $\psi=\chi\vee\phi$ and both $\chi$ and $\phi$ belong to $\mathtt{F}$;
\item $\psi=\neg\chi$ and $\chi$ belongs to $\mathtt{F}$;
\item $\psi$ is either $\top$ or $\bot$.
\end{enumerate}

We use the symbols
\begin{enumerate}
\item $\psi\implies\chi$ to denote $\neg \psi\vee\chi$;
\item $\psi\equiv\chi$ to denote $(\psi\implies\chi)\wedge(\chi\implies\psi)$.
\end{enumerate} 
\end{definition}

\begin{definition}\label{VALUA:def}
A valuation over a set of atoms $\mathcal{A}$ is a mapping $v$ from
$\mathcal{A}$ to the set $\{True, False\}$ and, as usual, we extend the mapping $v$ from a formula $\psi$ to by reducing the complexity of writing a
formula. Define the extension mapping $v$ to a mapping $v'$ from the formulas into the set $\{True, False\}$ as above,
\begin{enumerate}
\item If $\psi\in\mathcal{A}$, $v'(\psi)=True$ if $v(\psi)=True$;
\item $\psi=\chi\wedge\phi$, $v'(\psi)=True$ if $v'(\chi)=True$ {\em and}
  $v'(\phi)=True$;
\item $\psi=\chi\vee\phi$, $v'(\psi)=True$ if $v'(\chi)=True$
  {\em or} $v'(\phi)=True$;
\item $\psi=\neg\chi$, $v'(\psi)=False$ if and only if $v'(\chi)=True$;
\item $v'(\bot)=False$ and $v'(\top)=True$.
\end{enumerate}
If we do not have $v'(\psi)=True$, then $v'(\psi)=False$.
\end{definition}

\begin{definition}[SAT]\label{SAT:dfn}
A formula $\chi$ in Boolean Logic is said {\em satisfiable}\index{satisfiable}
if there is a valuation $v$ from the set $\{True, False\}$ onto the set
of atoms of $\chi$ so that  $v(\chi)$ is {\bf True}. If no such
valuation exists, that is $v(\chi)$ is {\bf False} for any valuation,
then we say that $\chi$ is {\em unsatisfiable}\index{unsatisfiable}.
\end{definition}

Any formula in the Boolean Logic is {\em decidable}, that is, given a
formula $\chi$, we can decide if $\chi$ is {\bf satisfiable} (True for some
valuation) or {\bf unsatisfiable} (False for any valuation). 
The complexity problem consists of deciding between the values, 
{\bf True} or {\bf False} in the most optimal way, see 
\cite{PAPA94}.

We use the symbols: $AND$, $OR$, $NOT$ and $IMPLIES$ to denote
external logical symbols.

\begin{definition}\label{CNN&MIN:dfn}
  A \twsat\ formula\index{\twsat\ formula} $\Psi$ is a conjunction of a finite
  number of  a disjunction of at most two literals. We write
\[
\Psi\equiv (l_{1}^{1}\vee l_{1}^{2})\wedge\dots\wedge
       (l_{\mathbf{s}}^{1}\vee l_{\mathbf{s}}^{2})\equiv
C_{1}\wedge\dots\wedge C_{\mathbf{s}}
\]

A subformula $C_{k}\equiv l_{k}^{1}\vee l_{k}^{2}$,
$1\leq k\leq \mathbf{s}$ of $\Psi$ is called a {\em clause}\index{clause}.

A \thsat\ formula\index{\thsat\ formula} $\Psi$ is a conjunction of a
finite number of 
a disjunction of at most three literals. We write
\[
\Psi\equiv (l_{1}^{1}\vee l_{1}^{2}\vee l_{1}^{3})\wedge\dots\wedge
       (l_{\mathbf{s}}^{1}\vee l_{\mathbf{s}}^{2}\vee l_{\mathbf{s}}^{5})\equiv
C_{1}\wedge\dots\wedge C_{\mathbf{s}}
\]

A subformula $C_{k}\equiv l_{k}^{1}\vee l_{k}^{2}\vee l_{k}^{3}$,
$1\leq k\leq \mathbf{s}$ of $\Psi$ is called a {\em clause}\index{clause}.
\end{definition}

Define pivoted \thsat\ formulas,
\begin{definition}\label{CARD:dfn}
A Pivoted \thsat\ formula\index{Pivoted \thsat\ formula} is a \thsat\ formula of the form,
\[
\begin{array}{l}
  (a_{1}\vee p_{1}^{1}\vee q_{1}^{1})\wedge\dots\wedge
  (a_{1}\vee p_{k_{a_{1}}}^{1}\vee q_{k_{a_{1}}}^{1})\wedge\\
  (\neg a_{1}\vee r_{1}^{1}\vee s_{1}^{1})\wedge\dots\wedge
  (\neg a_{1}\vee r_{k_{\neg a_{1}}}^{1}\vee s_{k_{\neg a_{1}}}^{1})\wedge\dots\wedge\\
  (a_{\mathbf{m}}\vee p_{1}^{\mathbf{m}}\vee q_{1}^{\mathbf{m}})\wedge\dots\wedge
 (a_{\mathbf{m}}\vee p_{k_{a_{\mathbf{m}}}}^{\mathbf{m}}\vee q_{k_{a_{\mathbf{m}}}}^{\mathbf{m}})\wedge\\
  (\neg a_{\mathbf{m}}\vee r_{1}^{\mathbf{m}}\vee s_{1}^{\mathbf{m}})\wedge\dots\wedge
  (\neg a_{\mathbf{m}}\vee r_{k_{a\neg _{\mathbf{m}}}}^{\mathbf{m}}\vee s_{k_{\neg a_{\mathbf{m}}}}^{\mathbf{m}}) 
\end{array}
\]
or, in factorized form,
\[
\begin{array}{l}
  (a_{1}\vee \big( (p_{1}^{1}\vee q_{1}^{1})\wedge\dots\wedge
  (p_{k{a_{1}}}^{1}\vee q_{k_{a_{1}}}^{1})\big)\wedge\\
  (\neg a_{1}\vee\big( (r_{1}^{1}\vee s_{1}^{1})\wedge\dots\wedge
  (r_{k_{\neg a_{1}}}^{1}\vee s_{k_{\neg a_{1}}}^{1})\big)\wedge\dots\wedge\\
  (a_{\mathbf{m}}\vee\big((p_{1}^{\mathbf{m}}\vee q_{1}^{\mathbf{m}}) \wedge\dots\wedge
 (p_{k_{a_{\mathbf{m}}}}^{\mathbf{m}}\vee q_{k_{a_{\mathbf{m}}}}^{\mathbf{m}})\big)
 \wedge\\
  (\neg a_{\mathbf{m}}\vee\big( (r_{1}^{\mathbf{m}}\vee s_{1}^{\mathbf{m}})
  \wedge\dots\wedge
  (r_{k_{a\neg _{\mathbf{m}}}}^{\mathbf{m}}\vee s_{k_{\neg a_{\mathbf{m}}}}^{\mathbf{m}}) \big)
\end{array}
\]
where the set of {\em pivots}\index{Pivot}, $\{a_{1},\neg a_{1},\dots,
a_{\mathbf{m}},\neg a_{\mathbf{m}}\}$ has no intersection with the set of entries\index{entries},
\[
\begin{array}{l}
\{p_{1}^{1},q_{1}^{1},\dots,p_{k{a_{1}}}^{1}, q_{k_{a_{1}}}^{1},
r_{1}^{1},s_{1}^{1},\dots,r_{k_{\neg a_{1}}}^{1},s_{k_{\neg a_{1}}}^{1},\dots,\\
p_{1}^{\mathbf{m}},q_{1}^{\mathbf{m}},p_{k_{\mathbf{m}}}^{\mathbf{m}},
q_{k_{\mathbf{m}}}^{\mathbf{m}},
r_{1}^{\mathbf{m}}, s_{1}^{\mathbf{m}},\dots,
r_{k_{\mathbf{m}}}^{\mathbf{m}}, s_{k_{\mathbf{m}}}^{\mathbf{m}}\}
\end{array}
\]

Denote the pivot $a_{i}$ by $1i$ and its conjugated pair $\neg a_{i}$ by $2i$

Call the set of entries of $a_{i}$\index{set of entries of $a_{i}$}
the set of conjunctions,
\[
(p_{1}^{i}\vee q_{1}^{i})\wedge\dots\wedge(p_{k{a_{1}}}^{i}\vee q_{k_{a_{1}}}^{i})
\]
and the set of entries\index{set of entries of $\neg a_{i}$} of $\neg a_{i}$ the set of conjunctions,
\[
(r_{1}^{i}\vee s_{1}^{i})\wedge\dots\wedge
(r_{k_{\neg a_{1}}}^{i}\vee s_{k_{\neg a_{1}}}^{i})
\]
\end{definition}

\begin{definition}\label{COMPL:dfn}
  Given a Pivoted formula $\Psi$, replace any clause of the form
  $a_{i}\vee p$, where $a_{i}$ is a pivot, by the clauses
  $(a_{i}\vee p\vee q_{p})\wedge(a_{i}\vee p\vee \neg q_{p})$, where the pairs
  of conjugated, $q_{p}$ and $\neg q_{p}$ are fresh new literals.

  A pivoted \thsat formula is said {\em complete}\index{complete \thsat} if
  all of its clauses contain three literals.
\end{definition}

\begin{lemma}\label{COMPL:lmm}
  Any pivoted \twsat\ formula is logically equivalent to its complete version.
\end{lemma}

From now on, we work with complete \thsat\ pivoted formulas.

It is well known that \twsat\ formulas can be polynomially solved, a
result originally shown in \cite{APT1979}. See \cite{PAPA94}. We will store
unsatisfiable groups of pivoted \thsat\ formulas in groups of \twsat.

\begin{proposition}\label{PAPA:prp}
  Given \twsat\ formula $\Psi$,
  \[
  (p_{1}\vee q_{1})\wedge\dots\wedge(p_{k}\vee q_{k})
  \]
  $\Psi$ is unsatisfiable if and only if it there is a pair of
  conjugated literals, $\{a,\neg a\}$ and a set of conjugated literals 
  \[
  \begin{array}{l}
 \{b_{1},\neg b_{1}\dots,b_{r},\neg b_{r},c_{1},\neg c_{1},\dots,c_{s},\neg c_{s}\}
  \end{array}
  \]
  so that
  \[
  (a\vee b_{1})\wedge(\neg b_{1}\vee b_{2})\wedge \dots\wedge(\neg b_{r}\vee a)
  \]
  and
  \[
  (\neg a\vee c_{1})\wedge(\neg c_{1}\vee c_{2})\wedge\dots\wedge
  (\neg c_{s}\vee\neg a)
  \]
 are subformulas of $\Psi$. 
\end{proposition}

Explore the above result as a support tool to achieve our claim.
We store the data contained in a pivoted \thsat\ formula in a cylindrical digraph.

\begin{definition}\label{CYL:dfn}
Let $\Psi$ be a pivoted (complete) \thsat\ formula. The
{\em cylindrical digraph} associated to $\mathcal{P}$\index{cylindrical
  digraph}, $\clndr=\langle V,E,label\rangle$\index{$\clndr=\langle V,E,label\rangle$}, where $V$ is a set of vertices, $E$ is a set of edges and
 $label$ is a mapping from $E$ to the set of parts of the set of entries is,
\begin{enumerate}
\item The set of vertices is the set of literals together with their negations;
\item If $a \vee b$ belongs to the set of entries
  $\{i_{1}j_{1},i_{2}j_{2},\dots,i_{k}j_{k}\}$, then,
  \begin{enumerate}
   \item $\neg a\Rightarrow b$ and $\neg b\Rightarrow a$ belong to the
     set edges, $E$; 
   \item $i_{1}j_{1},i_{2}j_{2},\dots,i_{k}j_{k}$ is the label of both edges
     $\neg a\Rightarrow b$ and $\neg b\Rightarrow a$.
  \end{enumerate}
\end{enumerate}
\end{definition}

Sequences of vertices connected by $\Rightarrow$ are a path, as we define below,
\begin{definition}\label{INTERVL:dfn}
Let  $a$ and $b$ be two literals, there is a path\index{path} from $a$
to $b$ if there is a subdigraph of
$\clndr$ the form
\begin{eqnarray}\label{SEQNCVS:eq}
  a\stackrel{l_{1}}{\Rightarrow} c^{1}\stackrel{l_{2}}{\Rightarrow} c^{2}
  \stackrel{l_{3}}{\Rightarrow}\dots
  \stackrel{l_{k-1}}{\Rightarrow} c^{k-1}
  \stackrel{l_{k}}{\Rightarrow}b
\end{eqnarray}
\end{definition}

Some subgraphs of the cylindrical digraph will be associated to combinations of
unsatisfiable formulas and we keep track of these subgraphs and, thus, we can
decide, in an optimal way whether they represent or not all possible unsatisfiable combinations.

We search for sequences of the form,
\begin{equation}\label{UNS:eqn}
\begin{split}
  \neg a\Rightarrow b^{1}\Rightarrow b^{2}\Rightarrow\dots\Rightarrow b^{t_{1}} \Rightarrow a\\
and\\
  a\Rightarrow c^{1}\Rightarrow c^{2}\Rightarrow\dots\Rightarrow c^{t_{2}}
  \Rightarrow \neg a
\end{split}
\end{equation}

We conclude that a Pivoted \thsat\ formula $\Psi$ has an unsatisfiable combination, with the entries chosen among the set of labels over sequences depicted in Equation \ref{UNS:eqn}.

\begin{notation}\label{N2M:ntt}
Let $\mathbf{2}^{\mathbf{m}}$, the set of all mappings from $\mathbf{m}$ to
$\mathbf{2}=\{1,2\}$. Write an element of $\mathbf{2}^{\mathbf{m}}$ as
$\{1j_{1},2j_{2},\dots,\mathbf{m}j_{\mathbf{m}}\}$, that is, the above
set represents the graphic of a mapping $f:\mathbf{m} \mapsto\mathbf{2}$ and an 
{\em entry}\index{entry}, $ij$ has the meaning that $f(i)=j$.

We identify an entry $a_{p}$ with $\mathbf{p}1$ and an entry $\neg a_{p}$
with $\mathbf{p}2$, $1\leq\mathbf{p}\leq\mathbf{m}$.
\end{notation}

\begin{definition}\label{FACT:dfn}
Given a Pivoted \thsat\ formula, $\mathcal{C}$, our task is to verify, in
the most efficient way, regarding the use of time and space, whether for
all all combinations $\sigma$ in  $\mathbf{2}^{\mathbf{m}}$, for all literals in
$\sigma$ regarded as $False$ whether the resulting $\twsat$\ formula is unsatisfiable, or, in a specular rephrasing,
we verify whether there is an element  
$\sigma=\{1j_{1},2j_{2},\dots,\mathbf{m}j_{\mathbf{m}}\}$
in $\mathbf{2}^{\mathbf{m}}$, in which the combination,
\begin{equation}\label{PC:eqn}
  \bigwedge_{1\leq r\leq \mathbf{m}}\{\wedge(p^{rj_{r}}\vee q^{rj_{r}})|
                              (p^{rj_{r}}\vee q^{rj_{r}})\in~contents~rj_{r}\}
\end{equation}
is satisfiable. 
\end{definition}

In short, we ask if the Pivoted \thsat\ formula is {\em unsatisfiable} or
{\em satisfiable}, by first making valuations ranging over the set of pivots.
Using Proposition \ref{PAPA:prp}, we search the sequences that lead to
unsatisfiable formulas and, thereafter, we decide whether we wrote all possible
combinations or not.

We set the keystones by defining the cylindrical digraph associated with a given
pivoted \sat\ formula. Any disjunction contained in an entry plays the role of a vertex in a cylindrical digraph.

\begin{definition}\label{IP:dfn}
Given a Pivoted Formula $\Psi$, an interval $[p,q]$\index{interval $[p,q]$} is the subdigraph that contains all sequences in between $p$ and $q$ and contains no loops, that is there  is not a vertex $r$ so that a sequence 
\[
  r\Rightarrow\dots\Rightarrow r
\] 
belongs to $[p,q]$ and the set of
{\em necessarily true literals}\index{necessarily true  literals},
$\mathcal{NEC}$, is the set of all pair of conjugated literals so that
the intervals $[\neg a,a]$ and $[a,\neg a]$ are non-empty.
\end{definition}

\begin{definition}\label{NEC:dfn}
  A pair of {\em necessarily true literals}\index{necessarily true literals}
  is a conjugated pair of literals, $\{a,\neg a\}$ so that there are paths
  from $a$ to $\neg a$ and from $\neg a$ to $a$.

  The set of necessarily true literals is denoted by
  $\mathcal{NEC}$\index{$\mathcal{NEC}$}.
\end{definition}

For each pair of conjugated literals, $\{a,\neg a\}\subseteq\mathcal{NEC}$,
we write distinct copies of $[\neg a,a]$ and $[a,\neg a]$. As usual, distinct copies are fabricated using Cartesian Products to differentiate intervals.

\begin{definition}\label{CSDDG:dfn}
  Given a cylindrical digraph, suppose that the cardinality of
  $\mathcal{NEC}$ is an even number $r\geq 2$ and the set of necessarily
  true literals is $\{p_{1},\neg p_{1},\dots,p_{r/2},\neg p_{r/2}\}$.
  Define the {\em set of closed
  digraph}\index{set of closed digraph}, $\clsddg$\index{$\clsddg$} as the
  union of all digraphs of the form
  $[\neg p_{l},p_{l}]\oplus [p_{l},\neg p_{l}]$\index{$[\neg p_{l},p_{l}]\oplus
  [p_{l},\neg p_{l}]$}, 
  
  \begin{enumerate}
  \item $[p_{l},\neg p_{l}]\times\{2l\}$, $1\leq l\leq r/2$;
  \item $[\neg p_{l},p_{l}]\times\{2l-1\}$, $1\leq l\leq r/2$.
  \end{enumerate}
  Define the edges and labels as the edges and labels inherited by the
  projection to the interval and endowed with the
  identification $(p_{2l-1},2l-1)\sim (p_{2l},2l)$
\end{definition}

\begin{definition}\label{CHAIN:dfn}
  A {\em chain}\index{chain} is a sequence contained in the set of closed
  digraphs,
\[
\begin{array}{l}
(\neg a,2l-1)\stackrel{l_{0}}{\Rightarrow}(c^{1},2l-1)\stackrel{l_{1}}{\Rightarrow}
(c^{2},2l-1)\stackrel{l_{2}}{\Rightarrow}\dots \stackrel{l_{k-1}}{\Rightarrow}
  (c^{k},2l-1)\stackrel{l_{k}}{\Rightarrow} \\(a,2l-1)\sim(a,2l)\\
  \stackrel{l_{k+1}}{\Rightarrow}(c^{k+1},2l)\stackrel{l_{k+2}}{\Rightarrow}
  (c^{k+2},2l) \stackrel{l_{k+2}}{\Rightarrow}\dots \stackrel{l_{t-1}}{\Rightarrow}
  (c^{t+1},2l)\stackrel{l_{t}}{\Rightarrow} (\neg a,2l)
  \end{array}
\]
The set of all chains\index{chains} is denoted by $\mathcal{P}$\index{$\mathcal{P}$} and called {\em Expansion of the Closed Digraph}\index{Expansion of the Closed Digraph}.

Consider the set of all chains in a closed digraph. Enumerate them by
$\{1,\dots,s\}$. To each chain $Ch_{t}$ enumerated by $t$, consider the
Cartesian Product $Ch_{t}\times\{t\}$, labels, edges inherited by the labels and edges in $Ch$. 
\end{definition}

\begin{definition}\label{P2CLSD:dfn}
  The mapping \index{$Ch\mathcal{P}\mapsto \clsddg$}
  \[
  \begin{array}{ll}
   Ch:&\mathcal{P}\mapsto \clsddg\\
      & (v,i)\rightarrow v
   \end{array}
  \]
  is the projection of vertices of $\mathcal{P}$ onto $\clsddg$.

  The {\em expansion to $\mathcal{P}$}\index{expansion to $\mathcal{P}$} is
  the fixed (exponential) writing of all maximal chains contained in a
  closed digraph.
  The enumeration of the chains is arbitrary and, from now on, remains fixed.

  Given the set of chains, if $(v,t)$ belongs to a chain in
  $\mathcal{P}$, then $v$ is a vertex in the cylindrical digraph and 
  $(u,v)$ is an edge of the  cylindrical digraph if there is a chain $Br_{t}$
  and two arrows $(u,t)\Rightarrow(v,t)$ in $Br_{t}$, in $Br_{t}$.
\end{definition}

\begin{definition}\label{P2CLSD2:dfn}
  Consider an edge $e=((u,t)\Rightarrow (v,t))$ in $\mathcal{P}$. Define
  $Ch(e)$\index{$Ch(e)$} as
  \[
    Ch(u,t)\Rightarrow Ch(v,t)
  \]
\end{definition}

\begin{definition}\label{CPT2:dfn}
  A {\em compatible set of entries}\index{compatible set of
  entries} is a set $U=\{i_{1}j_{1},\dots,i_{n}j_{n}\}$ so that
  for all pair $i_{r}$ and $i_{s}$, if $i_{r}=i_{s}$ then
  $j_{r}=j_{s}$. If a set of entries is not compatible, it is said an
  incompatible set.
  
  A set of edges $\mathbb{E}=\{e_{1},e_{2},\dots,e_{k}\}$
  is said {\em incompatible}\index{incompatible} if the union of it set of
  labels, $U=\{i_{1}j_{1},\dots,i_{n}j_{n}\}$ is an {\em incompatible set of
  entries}.
  If a set of edges is not incompatible, we say that $\mathbb{V}$ is
  {\em compatible}.
\end{definition} 

We wrote digraphs that encode all unsatisfiable choices. We must decide whether
we have, in fact, the set of all unsatisfiable choices. Again, we will not spell all the compatible, if any, choices. It is an expsize task.

Any chain $Br$ and a compatible choice in $Br$ is an unsatisfiable combination.
After writing all unsatisfiable combinations we must decide whether we wrote all possible combinations (we have an unsatisfiable formula) or not (the pivoted formula is satisfiable).

\begin{definition}\label{PATH:dfn}
Given a chain in $\mathcal{P}$,
\[
\begin{array}{ll}
\mathbf{p}= &c^{0}\stackrel{l_{1}}{\Rightarrow} c^{1}
  \stackrel{l_{2}}{\Rightarrow} c^{2}
  \stackrel{l_{3}}{\Rightarrow}\dots \stackrel{l_{k-1}}{\Rightarrow}
c^{k-1}\stackrel{l_{k}}{\Rightarrow} c^{k}
\end{array}
\]
a {\em compatible choice}\index{compatible choice} in $\mathbf{p}$ is
a set of compatible entries 
$\mathcal{E}=\{{j_{1}i_{1}},\dots,{j_{l}i_{l}}\}$,
so that for each edge $e_{i}=c^{i}\Rightarrow c^{i+1}$, $0\leq i<k$, there is an
entry $ij\in \mathcal{E}$ so that $ij\in label(e_{i})$.
\end{definition}

If one casts a tableau where all combinations of pivots, $\sigma\in\mathbf{2}^{\mathbf{m}}$ labeled as $False$ and there is a sequence
\[
\begin{array}{l}
  \neg a\stackrel{q_{0}}{\Rightarrow}c^{1}\stackrel{q_{1}}{\Rightarrow}
c^{2}\stackrel{q_{2}}{\Rightarrow}\dots \stackrel{q_{k-2}}{\Rightarrow}
c^{k-1}\stackrel{q_{k-1}}{\Rightarrow}\\
  a\stackrel{q_{k}}{\Rightarrow}c^{k+1}\stackrel{q_{k+1}}{\Rightarrow}
  c^{k+2} \stackrel{q_{k+2}}{\Rightarrow}\dots \stackrel{q_{t-2}}{\Rightarrow}
  c^{t-1}\stackrel{q_{t-1}}{\Rightarrow} \neg a=c^{t}
  \end{array}
\]
where
$\{q_{0},q_{1},q_{2},\dots,q_{k-1},q_{k},q_{k+1},\dots,q_{t-1},q_{t}\}\subseteq\sigma$ and each $q_{j}$ belongs to the label of the label $l_{j}$, $0\leq j\leq t-1$,
then the branch that falsifies the elements of $\sigma$ closes. So, the next question is whether all possible sequences appear in $\mathcal{P}$ and, therefore
the pivoted formula is unsatisfiable or not, that is, the pivoted formula is satisfiable.

Given a Pivoted \thsat\ formula $\Psi$, the set of all compatible choices in $\mathcal{P}$ is the set of all unsatisfiable combinations we will find in $\Psi$. We now ask how one can decide whether the set of all compatible choices entails all possible combinations, thus, the Pivoted \thsat\ is unsatisfiable.

Define the set of vertices {\em orthogonal} to the set of chains.

\begin{definition}\label{CHAINANTCHAIN:dfn}
  A set of edges $\tilde{seq}=\{e_{0},\dots,e_{n}\}$ is called an
  {\em antichain}\index{antichain} if for all $seq\in\mathcal{P}$,
  there is an element $e\in\tilde{seq}$ so that $e$ is an edge of $seq$.
\end{definition}

We do not focus on the set of vertices any more. Instead, focus on the set of {\em edges}, or, more specifically, on all the compatible choices of labels in $\mathcal{P}$. If all possible combinations were written, the formula would be unsatisfiable. Otherwise, the formula is satisfiable.
Do all possible choices encompass all the possible combinations?

\begin{definition}[ENTAILS]\label{ENTAILS:dfn}
Let $T(\mathbf{2}^{\mathbf{m}})$ be a set of partial mappings of $\mathbf{m}$
into $\{1,2\}$. We say that $T(\mathbf{2}^{\mathbf{m}})$
entails\index{entails} $\mathbf{2}^{\mathbf{m}}$ if for all  
$\gamma\in\mathbf{2}^{\mathbf{m}}$, there is an 
$\eta\in T(\mathbf{2}^{\mathbf{m}})$ so that $\eta$ is a restriction of
$\gamma$, that is, $\eta$ is a subset of some
$\gamma=\{i_{1}1,i_{2}2,\dots,i_{\mathbf{m}}\mathbf{m}\}$. 
\end{definition}

We search for chains in $\clsddg$. Writing all the paths is a
expsize long task. 
Recall that in Definition \ref{CSDDG:dfn}, $\mathcal{P}$ denotes the
set of all chains in $\clsddg$. 

\begin{definition}\label{CTTCOMB:dfn}
Let $\tau$ be the set of all entries
$\eta=\{1j_{1},\dots,\mathbf{m}j_{\mathbf{m}}\}$ contained in
$\mathbf{2}^{\mathbf{m}}$ so that
\[
\begin{array}{l}
\exists seq\in\mathcal{P}~~
\forall e\in seq~~ \exists 1\leq l\leq \mathbf{m}
(i_{l}j_{l}\in label (e)) 
\end{array}
\]
\end{definition}

Let $\tilde{\tau}$ be the complementary of $\tau$, that is, the set of all
entries $\tilde{\eta}=\{1j_{1},\dots,\mathbf{m}j_{\mathbf{m}}\}$ so that
\[
\begin{array}{l}
\forall seq\in\mathcal{P}~~
\exists e\in seq ~~\forall 1\leq l\leq \mathbf{m}
(i_{l}j_{l}\not\in label (e))
\end{array}
\]

Proposition \ref{SEARCHALL:prp} below marks a decisive step to build
our theoremhood.
After building the closed digraphs, and applying Proposition \ref{SEARCHALL:prp},
our search is focused on looking for compatible antichains. Again, we
emphasize that our answer is a plain output, {\bf YES} or {\bf NO},
without naming the compatible antichains if any.
In other words, we have a plain program and we expect an answer 
{\bf YES} or {\bf NO} regarding the possibility of the existence of some
kind of input.

\begin{proposition}\label{SEARCHALL:prp}
The following assertions are equivalent,
\begin{enumerate}
\item There is no compatible set of entries in $\tilde{\tau}$;
\item $\tau$ entails $\mathbf{2}^{\mathbf{m}}$;
\item There is no compatible antichain in $\mathcal{P}$.
\end{enumerate}
\end{proposition}
\begin{proof} Clearly $1$ and $2$ are equivalent. If $seq$ is a
   compatible set of entries in $\tilde{\tau}$, then,
   $seq\not\in{\tau}$ for the sets are complementary.

   1 and 2 imply 3. Suppose that there is a compatible antichain in
   $\mathcal{P}$, $Ach=\{e_{1},\dots,e_{k}\}$. Let 
   $lb(Ach)$ be the union of the set of labels of each vertex of $Ach$.
   As $Ach$ is compatible, $lb(Ach)$ has the form
   \[
   \{t_{1}r_{1}, t_{2}r_{2},\dots,t_{l}r_{l}\}
   \]
   where $1\leq t_{1}<t_{2}<...<t_{l}\leq\mathbf{m}$ and $r_{i}$ is either $1$
   or $2$. That is, for all $1\leq n\leq l$ we have only one choice, either
   $r_{n}= 1$ or $r_{n}=2$. So, for all  $1\leq s\leq\mathbf{m}$,
   we can choice a $\tilde{t}_{s}s$ so that $\tilde{t}_{s}\not\in ach$.

   Therefore, there is a non-empty set
   \[
   Comp=\{\tilde{t_{1}}1,\tilde{t_{2}}2,\dots,\tilde{t_{m}}\mathbf{m}\}
   \]
   that is compatible and $Comp\cap Ach=\emptyset$.
   Conclude that $Comp$ belongs to $\tilde{\tau}$ and, thus, 
   there is a compatible array in $\tilde{\tau}$.

   Suppose that 1 is false and
   $s=\{t_{1}r_{1}, t_{2}r_{2},\dots, t_{\mathbf{m}}r_{\mathbf{m}}\}$ is a
   compatible set of entries in $\tilde{\tau}$ and, then, for all sequence
   $s_{r} \in\mathcal{P}$, there is an edge $e_{r}$ so that
   $label(e_{r})\cap s=\emptyset$. Let
   $A=\{e_{r}|e_{r}\in s_{r}~AND~(label(e_{r})\cap s=\emptyset)\}$.
   As $s$ is compatible, $A$ is
   compatible, so, $A$ is compatible
   antichain in $\mathcal{P}$.
\end{proof}

Due to Theorem \ref{SEARCHALL:prp}, we can solve our game by deciding
whether all antichains are incompatible or if there is a
compatible antichain. Bear in mind that the word ``entails $2^{\mathbf{m}}$'' encodes the fact that all the combinations in the pivoted \thsat\ formula are unsatisfiable.
\section{The Search For Compatible Antichains}\label{LASTIDEA:sec}

In this section, we develop the tools we use for deciding whether a closed digraph entails $2^{\mathbf{m}}$. Using Proposition \ref{SEARCHALL:prp}, we search for compatible antichains in the closed digraph.
  
There is no need to search for unsatisfiable combinations. Using Proposition \ref{SEARCHALL:prp}, we 
search compatible antichains, that is, we deal with pivots. Sequences of vertices were used to build the closed digraph. By using Proposition \ref{SEARCHALL:prp}, our attention goes to the question of whether there is a compatible antichain in the closed digraphs and, thus, the set of closed digraphs do not entail $\mathbf{2}^{\mathbf{m}}$. Again, we stress that we do not develop tools to write out all the antichains and, therefore, spell out all the possible compatible combinations which is an expsize problem.
The polysize problem consists of the search for a plain output {\bf YES},
there are antichains or {\bf NO}, there is no antichain.

The set of all chains stores all possible
unsatisfiable combinations one can perform.
We already wrote the closed digraph and, now,
our attention is devoted to the search for compatible antichains, a search we perform in the set of labels.
Each edge of a digraph has a label that will be used to
point out the relevant parts that weigh in our search for antichains.

Given a closed digraph, we first modify the shape of the digraph to
avoid exponential branching, the subject of the first Subsection \ref{SDssc}.
Lastly, we show that the reshaping, the {\em Linearized Digraph} does not change
the set of compatible antichains.

In the second step, Subsection \ref{DCA:ssc}, after modifying the closed digraph
we write a digraph that marks the compatible antichains orthogonal to the chains of the Linearized Digraph. A result of an empty digraph signalizes no compatible antichains in the Linearized Digraphs and, as we show in this Section, no compatible antichain in $\mathcal{P}$.
Proposition \ref{SEARCHALL:prp} below marks a decisive step to build
our theoremhood.
After building the closed digraphs, and applying Proposition \ref{SEARCHALL:prp},
our search is focused on looking for compatible antichains. Again, we
emphasize that our answer is a plain output, {\bf YES} or {\bf NO},
without naming the compatible antichains, if any.
In other words, we have a plain program and we expect an answer 
{\bf YES} or {\bf NO} regarding the possibility of the existence of some
kind of input. 
\subsection{A Simpler Digraph}\label{SDssc}

Here, we model a new shape to the set of closed digraphs.
To avoid an exponential search for the antichains, we build a digraph, the
{\em Linearized Digraph} less complex
than the closed digraph that is a mirror of the antichains in the Expansion of the Closed Digraph. The search for compatible antichains is unchanged
because we show there are compatible antichains in the Linearized Digraph if and only if there are compatible antichains in the closed digraph. We show that the search for compatible antichains over the Linearized Digraph is polynomial in time and space.

In this section, we develop the tools for deciding whether a closed digraph entails $2^{\mathbf{m}}$. Using Proposition \ref{SEARCHALL:prp}, we search for compatible antichains in the closed digraph instead of searching for unsatisfiable combinations. The set of sequences of vertices, played the main role in building closed digraphs. Guided by Proposition \ref{SEARCHALL:prp}, our attention goes to the question of whether the labels contain a compatible antichain and, thus, the set of closed digraphs do not entail $\mathbf{2}^{\mathbf{m}}$. Again, we stress that we do not
develop tools to write out all the antichains and, therefore, spell out all the
possible compatible combinations, which is an expsize problem.
The polysize problem consists of the search for a plain output {\bf YES},
there are antichains or {\bf NO}, there is no antichain.

The set of all chains stores all possible
unsatisfiable combinations.

Given a closed digraph, we first modify the shape of the digraph to
avoid exponential branching, the subject of the first Subsection \ref{SDssc}. We must show that the reshaping, the {\em Linearized Digraph} does not change
the set of compatible antichains.

Our proof strategy, schematized to avoid exponential multiplication on writing the set of chains, is the transformation of closed digraphs into linear digraphs. 

\begin{definition}\label{RT:dfn}
  Given a closed digraph, define,
  \begin{enumerate}
  \item $\mathtt{R}$, the set of roots\index{roots} is  
    $\{a\in V|\nexists b\in V((a\Rightarrow b)\in E)\}$;
 \item $\mathtt{T}$, the set of tops\index{tops} is
    $\{a\in V|\nexists b\in V((b\Rightarrow a)\in E\}$;
 \item If there are at least two edges
    \[
    \begin{array}{l}
    a\Rightarrow b_{1}\\
    a\Rightarrow b_{2}
    \end{array}
    \]
    or
    \[
    \begin{array}{l}
    b_{1}\Rightarrow a\\
    b_{2}\Rightarrow a
    \end{array}
    \]
    where $b_{1}\neq b_{2}$ , then $a$ is a {\em branching}\index{branching}.
  \end{enumerate}
\end{definition}

\begin{definition}\label{MAXBRCHS:dfn}
A sequence of labeled edges in a closed digraph is called a 
{\em maximal branch}\index{maximal branch} if seq is of the form
\begin{eqnarray}\label{ALTSEQMAX:eq}
  c^{0}\stackrel{l_{1}}{\Rightarrow} c^{1}\stackrel{l_{2}}{\Rightarrow} c^{2}
  \stackrel{l_{3}}{\Rightarrow}c^{3}\dots c^{k-1}
  \stackrel{l_{k}}{\Rightarrow}c^{k}
\end{eqnarray}
where $c^{k}$ is either a root or a branching and $c^{0}$ is either a top or a
branching and no other branching belongs to the sequence.

If an interval $[a,b]$ contains no branching besides $a$ and $b$, then it is called a {\em linear interval}\index{linear interval}.

A linear interval $[a,b]$ where $b$ is a root, is called a {\em root interval}\index{root interval}.
\end{definition}

\begin{definition}\label{BRCHS:dfn}
  Use the term {\em branch}\index{\em branch} to design a linear sequence
  \[
  c^{0}\Rightarrow c^{1}\Rightarrow\dots\Rightarrow c^{k-1}\Rightarrow c^{k}
  \]
  where except by, perhaps, $c^{0}$ and $c^{k}$, no other vertex is a branching.
\end{definition}

We will define the operations over the closed digraphs, Lifting, Multiplication of Branches and the Addition of Labels, the operations,
\ref{DIFCTTN:dfn}, \ref{OPERATIONS2:dfn} and \ref{ADDLBL:dfn} and, without any
loss in our search for compatible antichains, the 
given closed digraph is reshaped into a
simpler digraph, the {\em Linearized Digraph}. Of course, we have to show that we do not lose or gain
information about compatible antichains.
Examples of the operations and more explanations can be found in the
institutional page,
\[
www.ime.usp.br/\!\sim\!\! weiss
\]
In addition, Example \ref{EX1:exp} illustrates the use of all operations defined next.

\begin{definition}[Lifting]\label{DIFCTTN:dfn}
  Given a closed digraph $Gr$,
  \begin{itemize}
  \item If $B$ and $R$ are two consecutive linear intervals, that is, they share a common branching, $v$,
\[
\begin{array}{l}
  br_{11}=v_{11}\Rightarrow \dots\Rightarrow v_{1t_{11}}\Rightarrow v\\
  \dots\\
  br_{1k}=v_{1k}\Rightarrow \dots\Rightarrow v_{kt_{1k}}\Rightarrow v
\end{array}
\]
and
\[
\begin{array}{l}
br_{21}=v\Rightarrow v_{21}\Rightarrow \dots\Rightarrow v_{2t_{21}}\\
  \dots\\
br_{2k}=v\Rightarrow v_{2k}\Rightarrow \dots\Rightarrow v_{2t_{2k}}
\end{array}
\]
the Lifting of $Gr$ by $B$ and $R$\index{Lifting by $B$ and $R$} is 
$Gr(B\star R)=\langle V_{B\star R}, E_{B\star R},label\rangle$,
\index{$Gr(B\star R)=\langle E_{B\star R},
  E_{B\star R},label\rangle$}
\[
\begin{array}{ll}
  V_{B\star R}=&
  (V_{G}\setminus\{v\})\cup(\{v\}\times\{1,\dots,k\}\\
  E_{B\star R}=&
       (E_{Gr}\setminus\{v_{1t_{11}}\Rightarrow v,\dots,v_{kt_{ik}}\Rightarrow v,\dots,v\Rightarrow v_{21},\!\dots\!,v\Rightarrow v_{2k}\})\\
    &\cup
  \{v_{1t_{11}}\Rightarrow (v,1),\dots,v_{kt_{1k}}\Rightarrow (v,k),\\
  & \dots,(v,1)\Rightarrow v_{21},\dots,(v,k)\Rightarrow v_{2k}\}
\end{array}
\]
The labels of each $w\Rightarrow(v,j)$ and $(v,j)\Rightarrow w$ are the same labels associated with $w\Rightarrow v$ and $v\Rightarrow w$, respectively.
\item If $R$ is of the form
\[
\begin{array}{l}
br_{1}=v_{1}\Rightarrow \dots\Rightarrow v_{t_{1}}\Rightarrow r\\
  \dots\\
br_{n}=v_{n}\Rightarrow \dots\Rightarrow v_{t_{n}}\Rightarrow r
\end{array}
\]
where $r$ is a root, the Lifting of $Gr$ by $R$\index{Lifting of $Gr$ by $R$} is
$Gr(R)=\langle V_{R},E_{R},label\rangle$\index{$Gr(R)$} given by,
\[
\begin{array}{ll}
V_{R}=&(V\setminus\{r\})\cup(\{r\}\times\{1,\dots,n\})\\
E_{R}=&(E\setminus\{v_{t_{1}}\Rightarrow r,\!\dots\!,v_{t_{n}}\Rightarrow r\})\cup\\
     &\{v_{t_{1}}\Rightarrow (r,1),\dots,v_{t_{n}}\Rightarrow (r,n)\}
\end{array}
\]
The labels of each $v_{t_{j}}\Rightarrow(r,j)$ are the same labels of $v_{t_{j}}\Rightarrow r,$.
\item If $B$ is of the form
\[
\begin{array}{l}
  br_{11}=a\Rightarrow v_{1}\Rightarrow\dots\Rightarrow v_{t_{1}}\\
  \dots\\
  br_{1k}=a\Rightarrow v_{m}\Rightarrow\dots\Rightarrow v_{t_{m}}
\end{array}
\]
where $a$ is a top, the Lifting of $Gr$ by $B$\index{Lifting of $Gr$ by $B$} is
$Gr(B)=\langle V_{B},E_{R},label\rangle$\index{$Gr(B)$} given by,
\[
\begin{array}{ll}
V_{B}=&(V\setminus\{a\})\cup(\{a\}\times\{1,\dots,m\})\\
E_{B}=&(E\setminus\{a\Rightarrow v_{1},\dots,a\Rightarrow v_{m}\})\cup\\
     &\{(a,1)\Rightarrow v_{1},\dots,(a,m)\Rightarrow v_{m}\}
\end{array}
\]
The labels of each $(a,j) \Rightarrow(v_{j})$ are the labels of
$a \Rightarrow v_{j}$ 
\end{itemize}
\end{definition}

\begin{definition}[Multiplication of a branch]\label{OPERATIONS2:dfn}
  Let
  \[
  br=v\Rightarrow v_{1}\Rightarrow\dots\Rightarrow v_{j}\Rightarrow u
  \]
  be a maximal branch contained in a linear interval
  $[v,u]=\langle V,E,label\rangle$,
  $\mathtt{m}=\{1,\dots,m\}$ a finite set and
  $M=\{v_{1},\dots,v_{j}\}\times\{1,\dots,m\}$.

  The {\em Multiplication of the branch $br$}\index{Multiplication of a
  branch $br$} is the the interval $[v,u]'$, where the vertices are given
  by the set $(V\setminus\{v_{1},\dots,v_{j}\})\cup M$ and the edges of
  $V\setminus\{v_{1},\dots,v_{j}\}$ plus the edges in the new sequences,
  \[
  v\Rightarrow (v_{1},i)\Rightarrow\dots\Rightarrow(v_{j},i)\Rightarrow u
  \]
  $1\leq i\leq m$.
\end{definition}

As one cannot form edges among vertices in distinct linear sequences, the definition is consistent.

\begin{definition}\label{UPDOWN:dfn}
  Given a closed digraph and a vertex $a$, define
  $Up_{a}$\index{$Up_{a}$} and $Down_{a}$\index{$Down_{a}$},
  respectively as sets of vertices,
  \[
  \begin{array}{l}
  \{v\in V|\exists n\in\mathbb{N} \exists\{v_{1},v_{2},\dots,v_{n}\}\subseteq V|
  (v\Rightarrow v_{1}\Rightarrow v_{2}\Rightarrow\dots\Rightarrow v_{n}
  \Rightarrow a\}\\
  \{w\in V| \exists n\in\mathbb{N} \exists\{w_{1},w_{2},\dots,w_{n}\}\subseteq V|
  (a\Rightarrow w_{n}\Rightarrow\dots\Rightarrow w_{2}\Rightarrow w_{1}
  \Rightarrow w\}
  \end{array}
  \]

  An edge $v_{1}\Rightarrow v_{2}$ belongs to $Up_{a}$ if $v_{1}$
  belongs to $Up_{a}$. An analogous definition is given to $Down_{a}$, that is,
  $v_{2}$ belongs to $Down_{a}$.
\end{definition}

Next, we present a tool to avoid the creation of new compatible antichains
in the linearized digraph we will build next as a simpler, easier to manage
compared to the closed digraph. Suppose there is no addition of new labels. In that case, the Linearized Digraph we build next as a simpler digraph will present new compatible antichains with no corresponding antichain under a suitable mapping in $\mathcal{P}$.

\begin{definition}[Adding a Label]\label{ADDLBL:dfn}
  Let the set of the branching in a closed digraph be $B=\{a_{1},\dots,a_{t}\}$.
  Consider a set of conjugated pairs of literals,
  $B=\{p_{a_{1}},p_{\neg a_{1}},\dots,p_{a_{t}},p_{\neg a_{t}}\}$ that are all distinct
  from the literals we have been using.
  The {\em addition of the new pair of literals to the set
  of branching}\index{addition of the new pair of literals to the set
  of branching} is the new set of labels, $label'$ defined by,
  \[
  \begin{array}{l}
    label'(e)=label(e)\cup
              \{p_{a_{i}}|e\in Down(a_{i})\}\cup
              \{\neg p_{a_{i}}|e\in Up(a_{i})\}
  \end{array}
  \]
\end{definition}

We do not overload our manuscript, so we drop the prime symbol, $label'$ and
write only $label$,

In Example \ref{EX1:exp}, we emphasize the role of adding a pair of conjugated
literals that work to prevent the creation of a new antichain that never
existed in the originally set $\mathcal{P}$.

To perform Lifting, Multiplication of Branches and Addition of a Label
to obtain a simpler digraph, we must show that there will be no prejudice in counting antichains. We cannot create or erase antichains.

Reshape a closed digraph into a set of linear intervals on Procedure \ref{CSD2LIN:prp}. The mainstay to
sustain our construction is that, after writing Procedure
\ref{CSD2LIN:prp}, we obtain a linear digraph that is equivalent
to the originally closed digraph. We show that our remodeling will not
impair the preexisting antichains or create new antichains.

Lifting, Multiplying Branches and Adding labels in a closed digraph are operations used to generate linear digraphs with no branching and, mainly, no loss of information about compatible antichains stored in $\mathcal{P}$ with no exponential multiplication of branches. The operations are illustrated in
Example \ref{EX1:exp}.

\begin{procedure}[Changing a Closed Digraph into a Linear
  Digraph]\label{CSD2LIN:prp}
  Given a closed digraph  $\clsddg$, and a set of conjugated literals distinct
  of any literals from the set of labels of $\clsddg$,
  $S=\{p_{1},\neg p_{1},\dots, p_{r},\neg p_{r}\}$, where
  $V$ is the cardinality of literals (vertices).
  First, we transform the closed digraph, by Multiplying branches,
  into a digraph whose intervals, from branching to a root, are
  maximal linear sequences.
  After we obtain a digraph whose intervals, from branching to a root, are
  maximal linear sequences, we multiply branches, add suitable labels,
  and perform lifting.
  \begin{enumerate}
  \item Let the roots of the closed digraph be
    $\{e_{1},\dots,e_{n}\}$ and let the set of all root intervals be
   \[
  \begin{array}{l}
    r_{11} \Rightarrow\dots\Rightarrow e_{1}\\
    \vdots\\
    r_{l_{1}1} \Rightarrow\dots\Rightarrow e_{1}\\
    \vdots\\
    r_{1s} \Rightarrow\dots\Rightarrow e_{n}\\
    \vdots\\
    r_{l_{s}} \Rightarrow\dots\Rightarrow e_{n}
  \end{array}
  \]
  lift each maximal branch. Obtain the digraph $Modf$,
  \[
    \begin{array}{l}
    r_{11} \Rightarrow\dots\Rightarrow(e_{1},1)\\
    \vdots\\
    r_{l_{1}1} \Rightarrow\dots\Rightarrow(e_{1},l_{1})\\
    \vdots\\
    r_{1s} \Rightarrow\dots\Rightarrow(e_{r},1)\\
    \vdots\\
    r_{l_{s}s} \Rightarrow\dots\Rightarrow(e_{r},l_{r})
  \end{array}
    \]
  \item Suppose we obtain sets of root intervals of the form,
  \[
    \begin{array}{l}
    u\Rightarrow\dots\Rightarrow v_{1}\\
    \vdots\\
    u\Rightarrow\dots\Rightarrow v_{n}
  \end{array}
    \]
  where $\{v_{1},\dots,v_{n}\}$ are roots in the digraph.
    
  Let all the maximal branches ending on $u$ be,
  \[
    \begin{array}{l}
    w_{1} \Rightarrow\dots\Rightarrow u\\
    \vdots\\
    w_{t} \Rightarrow\dots\Rightarrow u
  \end{array}
    \]
  Multiply the branches by $m=max\{n,t\}$ and lift.
  Add to the set of labels a pair of conjugated literals $p_{u}$ and
  $\neg p_{u}$ never used in our process, to, respectively the edges in
  $Down_{u}$ and $Up_{u}$. Rewrite $Modf$ as,
  \[
    \begin{array}{l}
      w_{1} \Rightarrow\dots\Rightarrow (u,1)\Rightarrow\dots
      \Rightarrow v_{1}\\
    \vdots\\
    w_{m} \Rightarrow\dots\Rightarrow (u,m)\Rightarrow\dots\Rightarrow v_{m}
  \end{array}
    \]
 Proceed until all intervals are lifted.   
\item Finally, we have sets of maximal branches of the form,
    \[
    \begin{array}{l}
     v \Rightarrow\dots\Rightarrow u\Rightarrow\dots\Rightarrow t_{1}\\
    \vdots\\
    v \Rightarrow\dots\Rightarrow u\Rightarrow\dots\Rightarrow t_{o}
  \end{array}
  \]
  where $v$ is a top and $\{t_{1},\dots,t_{o}\}$ are roots. Lift all the branches
  and obtain
   \[
    \begin{array}{l}
     (v,1) \Rightarrow\dots\Rightarrow u\Rightarrow\dots\Rightarrow t_{1}\\
    \vdots\\
    (v,o) \Rightarrow\dots\Rightarrow u\Rightarrow\dots\Rightarrow t_{o}
  \end{array}
  \]
\end{enumerate}
  Call the  {\em Linearized Closed Digraph}\index{Linearized Closed Digraph}
  to the linear digraphs obtained after the successive applications of
  Lifting, Multiplication and Addition of a label Denote the Linearized
  Closed Digraph by $Linclsd$\index{$Linclsd$}.

  Suppose that each linear branch has the cardinality $k$ and write
  each linear branch in  $Linclsd$ as $Br_{i}$ as a short for
  $Br_{i}\times\{i\}$. Define the edges whose arrow is inherited by
  $v\Rightarrow u$.
\end{procedure}

We show that a Linearized Digraph, despite being simpler than the set of chains, preserves compatible antichains in the sense that
we do not create nor erase information about compatible antichains in
Lemma \ref{PROJ:lmm}. Recall that the vertices of either $\mathcal{P}$ and $Linclsd$ are pairs, a vertex in the closed digraph and a number.

\begin{definition}\label{VERTEXID:dfn}
  The mapping \index{$Proj:Linclsd\rightarrow\clsddg$}
  \[
  \begin{array}{ll}
    Proj:&Linclsd\rightarrow\clsddg\\
         &(v,j)\mapsto v
   \end{array}
  \]
  is the projection of vertices of $Linclsd$ onto $\clsddg$.
\end{definition}
    
\begin{definition}\label{PROJECEDGES:dfn}
 Consider the edge $e'=(u'\Rightarrow v')$ in $Linclsd$. Define
 $Proj(e')$\index{$Proj(e')$}  as the edge,
 \[
   Proj(u')\Rightarrow Proj(v')
 \]
\end{definition}

Definitions \ref{PROJECEDGES:dfn} and \ref{P2CLSD2:dfn}
are well posed for if
$u'\Rightarrow v'$ and $u''\Rightarrow v''$ are two edges, respectively in $Linclsd$ and in $\mathcal{P}$, then
$Proj(u')\Rightarrow Proj(v')$ and $Ch(u'')\Rightarrow Ch(v'')$ are, equally, edges in $\clsddg$.

\begin{lemma}\label{BIJ$ProjREL:lmm}
  Let $ach$ be an antichain in the Linearized Digraph. Then, for all
  branching $a$, either $Proj(Ach)\cap Up(a)$ or $Proj(Ach)\cap Down(a)$
  is empty.
\end{lemma}
\begin{proof} Due to the addition of labels, there is no branching $a\in\clsddg$ so that two edgers $e_{1}\in Up(a)$ and $e_{2}\in Down(a)$ belong to a compatible antichain $ach$ in the Linearized Digraph because both are incompatible. 
\end{proof}

\begin{definition}
  The set of all antichains in $\mathcal{P}$ is denoted by
  $\Sigma$\index{$\Sigma$, set of all antichains in $\mathcal{P}$} and the
  set of all antichains in $Linclsd$ is
  denoted by $\Theta$\index{$\Theta$}.
\end{definition}


Have we created or erased antichains? As the {\em relevant antichains} are kept, we search for compatible antichains in $Linclsd$. We
show that there are compatible antichains in $Linclsd$ if and only if there are compatible antichains in $\mathcal{P}$.

Given a branching $a$, recall the definition of a branch in $Up(a)$ is a linear
sequence, in $\mathcal{P}$ of the form,
\[
v_{1} \Rightarrow v_{2}\Rightarrow\dots\Rightarrow a\\
\]

\begin{definition}
  For all set of edges $Z$ in $\mathcal{P}$ and branching $a$, we say that
  $Z$ represents ${Up}(a)$ if for all branching $Br$ in $Up(a)$, there is a
  $z\in Ch(Z)$ so that $z$ belongs to $Br$.
  A symmetric definition applies to ${Down(a)}$, that is, for all branching
  $Br$ in $Down(a)$, there is a $z\in Ch(Z)$ so that $z$ belongs to $Br$.
  
  For all set of edges $Z$ in $\mathcal{P}$, $Z$ involves\index{involves}
  a set of branching $\mu$ if for all branching $a\in\mu$, either $Z$
  represents ${Up}(a)$ or $Z$ represents ${Down(a)}$.
\end{definition}

\begin{proposition}\label{REPRESONE:prp}
  Given an antichain $\sigma$ in $\mathcal{P}$, for all branching $a$,
  $\sigma$ represents ${Up}(a)$ or ${Down(a)}$.
\end{proposition}
\begin{proof}
  Let $\sigma$ be an antichain in $\mathcal{P}$. The antichain
  ${Up}(a)$ or ${Down(a)}$ where ``or'' is not exclusive.

  Suppose otherwise.
  Let the branches of $Up(a)$ and $Down(a)$ as, respectively,
  \[
  \begin{array}{lll}
  u_{11}   &\dots & u_{r1}\\
  \vdots  &\dots &\vdots\\
  u_{1j_{1}}&\dots & u_{rj_{r}}\\
  a&\dots& a
  \end{array}
  \]
  and
  \[
  \begin{array}{lll}
  a& \dots& a\\
  v_{11}   &\dots & v_{t1}\\
  \vdots  &\dots &\vdots\\
  v_{1j_{1}}&\dots & v_{tj_{t}}
  \end{array}
  \]
  Combine the above chains and obtain the chains in $\mathcal{P}$,
\begin{equation}\label{CHAIN2ND:eqn}
  \begin{array}{l@{\hskip 5pt}l@{\hskip 5pt}l@{\hskip 5pt}l@{\hskip 5pt}l@{\hskip 5pt}l@{\hskip 5pt}l}
  u_{11}  &\dots&u_{11}       &\dots & u_{r1}&\dots &u_{r1} \\
  \vdots &\dots&\vdots      &\dots &\vdots&\dots &\vdots\\
  u_{1j_{1}} &\dots&u_{1j_{1}}  &\dots & u_{rj_{r}}&\dots &u_{rj_{r}}\\
  (a,1,1)&\dots&(a,1,t)&\dots&(a,r,1)&\dots&(a,r,1)\\
  v_{11}   &\dots& v_{t1}     &\dots &v_{11}&\dots &  v_{t1}\\
  \vdots  &\dots &\vdots    &\dots &\vdots&\dots&\vdots \\
  v_{1j_{1}}&\dots & v_{tj_{t}} &\dots &v_{1j_{1}}&\dots & v_{tj_{t}}
  \end{array}
 \end{equation}
 

  If the antichain $\sigma$ has no edges in columns, say $r$ and $s$ in,
  respectively $Up(a)$ and $Down(a)$, then the combination of the two columns
  $r\times s$, does not contain any edge of the antichain.

  In conclusion, for any antichain $\sigma$, for any branching $a$,
  either all edges of $Ch(\sigma)$ are a member of one edge in $Up(a)$ or
  all edges of $Ch(\sigma)$ belong to $Down(a)$.
\end{proof}

\begin{proposition}\label{REPRESALL:prp}
  Given an antichain $\sigma$ in $\mathcal{P}$, $\sigma$ {\em involves}
  the set of branching of $\mathcal{P}$.
\end{proposition}
\begin{proof}
  Let $a_{1}<\dots<a_{s}$ be an arbitrary ordering of all branches in the
  closed digraph.

  Use Proposition \ref{REPRESONE:prp} to $a_{1}$. If $\sigma$ represents
  $Up(a_{1})$, obtain a new antichain $\sigma_{1}$ by defining,
  \begin{enumerate}\label{UP1:itm}
  \item Let $Br_{11},\dots, Br_{1s_{1}}$ be the set of branches in
    $Up(a_{1})$, still in the closed digraph. For all $1\leq i\leq s_{1}$,
    choose a $x_{i}\rightarrow y_{i}\in Br_{1i}\cap Ch(\sigma)$.
    Erase any edge whose projection via $Ch$ in $Br_{1i}$ is different from
    $x_{i}\rightarrow y_{i}$ and add to $\sigma_{1}$ all
    $(x_{i},t)\rightarrow (y_{i},t)$ whose projection is $x_{i}\rightarrow y_{i}$.
  \end{enumerate}

  Otherwise, $\sigma$ represents ${Down(a_{1})}$ and proceed similarly.
  
    Obtain a new antichain $\sigma_{1}$. Notice that
    $Ach(\sigma_{1})\subseteq Ach(\sigma)$. Suppose that we obtained
    $\sigma_{i}$, $1\leq i\leq t<r$ and that
    $Ach(\sigma_{t})\subseteq\dots\subseteq
    Ach(\sigma_{1})\subseteq Ach(\sigma)$ and that either,
    \begin{enumerate}\label{UP1:itm}
    \item $Ach(\sigma_{l})\cap Up(a_{1})=\emptyset$ and
      $Ach(\sigma_{l})\cap Down(a_{1})$ has a singleton edge at each branch in
      $Down(a_{1})$ or,
    \item $Ach(\sigma_{l})\cap Up(a_{1})$ has a singleton edge at each
      branch in $Up(a_{1})$
      and $Ach(\sigma_{l})\cap Down(a_{1})=\emptyset$.
  \end{enumerate}
    
    Suppose that $\sigma_{t}$ represents ${Up(a_{t+1})}$.
 \begin{enumerate}\label{UP2:itm}
 \item Let $Br_{t+11},\dots, Br_{t+1s_{t+1}}$ be the branches in
    $Up(a_{t+1})$. For all $1\leq i\leq s_{t+1}$,
    choose a $x_{i}\rightarrow y_{i}\in Br_{t+1i}\cap Ch(\sigma_{t})$.
    Erase any edge whose projection via $Ch$ in $Br_{1t+i}$ is different from
    $x_{i}\rightarrow y_{i}$ and add to $\sigma_{t+1}$ all
    $(x_{i},t)\rightarrow (y_{i},t)$ whose projection is $x_{i}\rightarrow y_{i}$.
  \end{enumerate}
  
 Otherwise, $\sigma_{t}$ represents ${Down(a_{t+1})}$ and a mirror
 reasoning applies.

 As a result, obtain a series of nested antichain that we build (can build
 in several random ways),
 $Ach(\sigma)\supseteq\dots\supseteq Ach(\sigma_{t})$,  where the cardinality
 of $t$ does not exceed the cardinality of the set of branching. Given the
 already chosen ordering, for all branching $a_{i}$, we choose preferentially
 $Up(a_{i})$, if possible, and, if we do not have this choice, we have the
 choice $Down(a_{i})$.
 
  We have that $Ach(\sigma_{t})$ cannot be reduced and, using Proposition
  \ref{REPRESONE:prp} represents any branching in the set of branches, thus,
  $\sigma$, as well as its reduction, involves the set of branching.
\end{proof}

\begin{lemma}\label{PROJ:lmm}
  There is an onto mapping $Eq$ from $\Sigma$ to $\Theta$.
\end{lemma}
\begin{proof}
  The operations of Lifting, Multiplying and Adding Labels are fixed, so,
  after reducing an antichain $\sigma$ to its minimal version $\sigma_{t}$,
  we obtain, in the Linearized Digraph, an antichain $\overline{\sigma_{t}}$,
  the result of the fixed operation over $\sigma_{t}$ in the Linearized Digraph.
  The onto mapping is given by $Eq(\sigma)$ is the expansion of
  $Ch({\sigma_{t}})$.
\end{proof}

In short, our quest is whether there is a compatible antichain in a closed
digraph is equivalent to the quest  whether there is a compatible antichain
in a linearized digraph.
\subsection{Digraph of Compatible Antichains}\label{DCA:ssc}

Now, our attention is devoted to the search for compatible antichains, which we
perform in the set of labels.
Each edge of a digraph has a label that will be used to
point out the relevant parts that weigh in our search for antichains.

In outline, once we have linear digraphs, $k$ linear digraphs, each linear digraph in the column $r$ endowed $l_{r}$ vertices. To decide the existence of compatible antichains in the Linearized Digraph, we filter series of digraphs associated to each vertex $v$. These digraphs, called {\em nested digraph} will mark compatible antichains. 
In the world of Nested Digraphs, the set of labels plays the role of vertices.
To avoid a heavy load, we discard unnecessary elements. Labels are the fundamental elements and, as we advance our search, labels play the role of vertices and the connection if given by compatibility.

Up to now, we no longer need the set of literals, that played the role of vertices in the Linearized Digraph and whose sequences encode unsatisfiable combinations.
Once we cleaned new paths, we pursued for new strategy for solving a card. The maximal chains in the Closed Digraph were used to mark the unsatisfiable combinations and our search is centered on the search
for a compatible antichain. So, we are dealing only with the set of
labels.

It is time to discharge vertices and edges in the Linearized Digraph because
these elements were used to mark unsatisfiable choices and we no longer use
them. We use lighter gadgets ruling out the set of vertices, the set of
literals.

A sequence of linear digraphs of the form,
\[
a\stackrel{l_{1}}{\Rightarrow}b_{1}\stackrel{l_{2}}{\Rightarrow}\dots
\stackrel{l_{t}}{\Rightarrow} b_{t}\stackrel{l_{t+1}}{\Rightarrow}c
\]
corresponds to the sequence
\[
l_{1}\rightarrow l_{2}\rightarrow\dots\rightarrow
l_{t}\rightarrow l_{t+1}
\] 

After modifying the closed digraph
we write a digraph that marks the compatible antichains orthogonal to the chains of the Linearized Digraph. A result of an empty digraph signalizes no compatible antichains in the Linearized Digraphs and, as we show in this Section, no compatible antichain in $\mathcal{P}$.

Let $(r,s)$ be a enumeration of the pair {\em column} and {\em location}
in the array of $k$ columns each with $m_{k}$ elements. Write the pair
edge together with the enumeration. An edge
$(p,i,j)\Rightarrow(q,i,j+1)\Rightarrow(r,i,j+2)$ belongs to column $i$ and
the vertices belong to, respectively, the $j$, $j+1$ and $j+2$ places.

\begin{definition}\label{NEWDIG:dfn}
  Given a Linearized Digraph, define its
  {\em Companion Digraph}\index{Companion Digraph} as the digraph whose
  vertices are the set of indexes associated with each edge of the linearized
  digraph. Define an edge $(l_{1},i,j)\rightarrow (l_{2},i,j+1)$ if there
  is a sequence
  $(p,i,j)\stackrel{l_{1}}{\Rightarrow}(q,i,j+1)
  \stackrel{l_{2}}{\Rightarrow}(r,i,j+2)$ in the linearized digraph.
\end{definition}

If  $(p,i,j)\stackrel{l_{1}}{\Rightarrow}(q,i,j+1)
\stackrel{l_{2}}{\Rightarrow}(r,i,j+2)$
and $(p',i',j')\stackrel{l_{1}}{\Rightarrow}(q',i',j'+1)
\stackrel{l_{2}}{\Rightarrow}(r',i',j'+2)$
we have two distinct edges
$(l_{1},i,j)\rightarrow(l_{2},i,j+1)$ and
$(l_{1},i',j')\rightarrow(l_{2},i',j'+1)$, spite labels are the same.

We will not overload our manuscript. Distinct edges are originated from
distinct columns and places. From now on we omit all extra notation and drop
the pair column location.

Now, set the scenario for building nested digraphs associated with each vertex (a label). We expect that a $k^{th}$-interaction to generate an output {\bf satisfiable}. We set the necessary definitions.

\begin{definition}\label{LENGTH:dfn}
  Let $Lab$ be a Digraph. We say that $Lab$ has dept\index{dept} $k$ if
  the length of all maximal sequences of $Lab$ is $k$.
\end{definition}

\begin{definition}\label{NEWLABEL:dfn}
  Let $Lab$ be a Digraph and $v$ a vertex of $Lab$. An edge-labeled
  digraph for $v$\index{edge-labeled digraph for $v$} is a closed rooted
  digraph whose edges are labeled and whose root is $v$. We label the
  edges with a set of vertices, that is, there is a mapping from the set
  of edges of $Gr(v)$ into sets of vertices.
  \[
  G(v)=\langle V_{v},E_{v},labels_{v}: E_{v}\mapsto \mathcal{P}(V_{v})\rangle
  \]
  The mapping $labels_{v}$ associate to each edge a set of vertices and
  $labels(v)$ is called the {\em label of the vertex $v$}\index{label of the
  vertex $v$}.

  We require, moreover, that all maximal sequences in $G$
  have length $p$.
\end{definition}

Define the basic operations over {\em labeled digraphs}
we perform to obtain a simpler digraph and, without altering the
the state of the closed digraph concerning the question there are compatible
antichains or not.

\begin{definition}\label{UNION:dfn}
Given two labeled digraphs whose root is $v$,
\[
Gr_{v}=\langle V_{v},E_{v},labels {v}:V_{v}\mapsto\mathcal{P}(E_{v})\rangle
\]
and
\[
Gr_{v}'=\langle V_{v}',E_{v}',labels {v}':E_{v}'\mapsto\mathcal{P}(E_{v})\rangle
\]
their union\index{union of digraphs}, $Gr_{v}\cup Gr_{v}'$ is given by
\[
\langle V\cup V',E\cup E',labels''_{v}:E\cup E'\mapsto
\mathcal{P}(V\cup V')\rangle
\]
where $labels''_{v}:E\cup E'\mapsto\mathcal{P}(V\cup V')$ is given by
\[
\begin{array}{l}
labels''_{v}=label(v),~{\rm if} ~v\in E_{v}\setminus E_{v}'\\
labels''_{v}(v)=label'(v),~{\rm if} ~v\in E_{v}'\setminus E_{v}\\
labels''_{v}(v)=label(v)\cup label'(v),~{\rm if}, v\in E_{v}\cap E_{v}'
\end{array}
\]
Their intersection\index{intersection of digraphs, $Gr_{v}\cap Gr'_{v}$}
is given by,
\[
\langle E_{v}\cap E_{v}',E_{v}\cap E_{v}',labels'':E_{v}\cap E'_{v}
\mapsto\mathcal{P}(E_{v}\cap E_{v}')\rangle
\]
where $labels''$ is given by
$labels''(v)=label(v)\cap label'(v)$, if $v\in E_{v}\cap E_{v}'$.
\end{definition}

\begin{definition}[Nested Digraph]\label{NESTEF:dfn}
A rooted edge-labeled digraph,
$Gr_{v}=\langle V_{v},E_{v},labels {v}:E_{v}\mapsto\mathcal{P}(E_{v})\rangle$
is a {\em nested digraph of length $p$}\index{nested digraph of length $p$}
if $Gr_{v}$ is the union of a number $s$ of linear digraphs of a fixed
dept $p\geq 1$, 
\[
v\leftarrow a^{2}_{i}\leftarrow a^{3}_{i}\leftarrow\dots\leftarrow
a^{p}_{i}
\]
so that for each linear digraph, we have
$label(a^{i-1}_{i}\leftarrow a^{i}_{i})=\{a^{p}_{i},a^{p-1}_{i},\dots,a^{i}_{i}\}$,
$1\leq i\leq s$.

Notice that not necessarily vertices at each level $j$, $1\leq j\leq p$
are distinct but, at each distinct level, the vertex are pairwise distinct.

A {\em compatible nested digraph of length $p$}\index{compatible nested
digraph of length $p$} is the union of linear digraphs of length $p$ whose labels are compatible sets.
\end{definition}

\begin{lemma}\label{UNINEST:lmm}
The union of two nested digraphs of depth $p$, $G(v)$ and $G(v)'$
is a nested digraph.
\end{lemma}

The intersection of two nested digraphs of depth $p$ is not necessarily a
nested digraph but we can select {\em the} maximal nested digraph contained in
the intersection. Uniqueness follows from the fact that, given two nested
maximal digraphs, $G(v)$ and $G(v)'$, then the union of a nested digraph is nested and, therefore, we cannot have
two distinct maximal nested digraphs contained in $G(v)\cap G(v)'$.

The maximal nested digraph contained in the intersection of two nested
digraphs,  $G(v)$ and $G(v)'$, contains the set of all compatible
antichains contained in both $G(v)$ and $G(v)'$.

\begin{definition}\label{MAX:dfn}
  Let $G(v)$ and $G(v)'$ be two nested digraphs of depth $p$.
  Denote by $maxG(v)$\index{$maxG(v)$}
  the Maximal Nested Subdigraph contained in $G(v)\cap G(v)'$.
\end{definition}

We will write an algorithm, Pseudocode \ref{INTER:sch} to obtain $maxG(v)$, the maximal nested digraph contained in the intersection of two nested digraphs. 

\begin{pseudocode}\label{INTER:sch}
  %
  Given a digraph $G(v)=\langle V,E,label\rangle$, intersection of two
  nested digraphs, obtain, using the two pesudocodes above, $maxG(v)$.

  We denote the vertices of the intersection by $v^{m}_{i}$, where $m$ is a
  level and $i$ is the position of the vertex in the level $m$ and $v=v^{0}_{1}$.

  {\bf First Pseudocode:} From $m:=0,~step~1,~to~k-2$, do:\\
  For all edges
  $u^{m}_{i}\leftarrow v^{m+1}_{j}$, consider the vertices that form an ascending
  sequence of length $2$ with $u^{m}_{i}\leftarrow v^{m+1}_{j}$, say,  the set
  $\{w_{1}^{m+2},\dots,w_{l}^{m+2}\}$.
  For $s$ ranging from $1$ to $l$, consider as provisional labels,
  \[
  lbl_{i,j,s}(w_{s}^{m+2}\rightarrow v^{m+1}_{j})=
  label(w^{m+2}_{s}\rightarrow v_{j}^{m+1})\cap
  label(u^{m+1}_{j}\rightarrow v^{m}_{i})
  \]
  Define,
  \[
  label'(u^{m+1}_{j}\rightarrow v^{m}_{i})=
  \cup_{s}lbl_{i,j,s}(w_{l}^{m+2}\rightarrow v^{m+1}_{j})\cup\{v^{m+1}_{j}\}
  \]

  Once  we defined all labels, $label'(u^{m+1}_{j}\rightarrow v^{m}_{i})$, define
  \[
  label(w^{m+2}_{s}\rightarrow v_{j}^{m+1})=
  \cup_{i}lbl_{i,j,s}(w_{s}^{m+2}\rightarrow v^{m+1}_{j})
  \]  
  and, finally, $label(u^{m+1}_{j}\rightarrow v^{m}_{i})
  \gets label'(u^{m+1}_{j}\rightarrow v^{m}_{i})$.

  The final result of applying the first pseudocode is the union of linear
  digraphs ordered by $\supseteq$.

  {\bf Second Pseudocode:} We search
  backwards, for $m:=k~STEP~ -1~UNTIL~ 2$.

  For all $v^{m}_{i}$, let $L_{mi}$ be the set of all edges
  \[
  \{v^{m-1}_{j}|(u^{m}_{i}\rightarrow v^{m-1}_{j})~ so~ that~
  u^{m}_{i}\in label(u^{m}_{i}\rightarrow v^{m-1}_{j})
  \]
  For all $v^{m-1}_{j}\in L_{mi}$, define $G_{u^{m}_{i}v^{m-1}_{j}}$ as the union of
  all linear sequences
  \[
  w^{k}_{r_{k}}\rightarrow w^{k-1}_{r_{k-1}}\rightarrow\dots\rightarrow
  u^{m}_{i}\rightarrow v^{m-1}_{j}\rightarrow\dots\rightarrow v
  \]
  where $\{u^{m}_{i},v^{m-1}_{j}\}$ is contained in all labels of degree
  lower than $m-1$.
  
  For all $e$ in $Gr_{u^{m}_{i}v^{m-1}_{j}}$, define,
  \[
  label(e)=(label(e)\setminus\{u^{m}_{p},v^{m-1}_{q}|p\neq i~AND~q\neq j\}
  \]

  Define $G_{m,m-1}(v)$ as
   $\cup_{1\leq i\leq l_{m}}(\cup_{v^{m-1}_{j}\in L_{mi}}Gr_{u^{m}_{i}v^{m-1}_{j}})$

  Go to let $G\gets G_{m,m-1}(v)$ and perform step $m-1$. Repeat until level $2$.

  In the second procedure, we rule out all edges whose labels do not
  comply with the rules of nested digraphs. 
\end{pseudocode}

We show that we obtain a nested digraph that is the maximal nested digraph
contained in the intersection of two nested digraphs that share the same root.

\begin{proposition}[$maxG(v)$]\label{MAXINT:prp}
  Using Pseudocodes \ref{INTER:sch}, given two nested digraphs, $G(v)'$ and
  $G(v)''$, obtain $maxG(v)$.
\end{proposition}
\begin{proof}
  Define,
  \[
  \begin{array}{l}
  lbl_{i,j,s}(w_{s}^{m+2}\rightarrow v^{m+1}_{j})=
  label(w^{m+2}_{s}\rightarrow v_{j}^{m+1})\cap
  label(u^{m+1}_{j}\rightarrow v^{m}_{i})\\
   label(w^{m+2}_{s}\rightarrow v_{j}^{m+1})=
  \cup_{i}lbl_{i,j,m}(w_{s}^{m+2}\rightarrow v^{m+1}_{j})\\
  label(u^{m+1}_{j}\rightarrow v^{m}_{i})=
  \cup_{l}lbl(w_{l}^{m+2}\rightarrow v^{m+1}_{j})\cup\{v^{m+1}_{j}\}
  \end{array}
  \]

  Deduce from the previous sentences that
  $label(u^{m+1}_{j}\rightarrow v^{m}_{i})$ is
  the union, ranging over all $s$, of
  $lbl(w_{l}^{m+2}\rightarrow v^{m+1}_{j})$
  and $label(w^{m+2}_{s}\rightarrow v_{j}^{m+1})$ is the union of all temporally
  labels $lbl(w_{l}^{m+2}\rightarrow v^{m+1}_{j})$ and, thus, we obtain, using
  the first algorithm, the union of linear digraphs whose
  labels are ordered by $\subseteq$. 

  Using the second pseudocode, each $G_{u^{m}_{i}v^{m-1}_{j}}$ is built after
  a selection of all linear chains so that $\{u^{m}_{i},v^{m-1}_{j}\}$ is
  contained in the label of all edges of dept lower than $m-1$.
  $G_{u^{m}_{i}v^{m-1}_{j}}$ is modified by ruling out all labels
  
  \[
  \{w^{r}_{s}|(r=m~IMPLIES~s\neq i)~AND~(r=m-1~IMPLIES~s\neq j\}
  \]
  In this way, we keep only the linear arrays of depth $k$ so that it its
  edges of length less than $m$ contain $\{u^{m}_{i},v^{m-1}_{j}\}$ in their label,
  and, we rule out sequences that are not maximal and, thus, fulfill nested
  digraphs requirements.

  Define $G(v)$ as the union of all $G_{u^{m}_{i}v^{m-1}_{j}}$. Conclude that, as
  we work backward, from $k=m$ step $-1$ that we obtain a nested digraph.
\end{proof}

The search for $maxG(v)$ is polynomially bounded, as we show in the analysis of computational boundaries in
Proposition \ref{MAXNSTD:prp}.

After writing the {\em Linearized Digraph}, we filter all
compatible antichains. We write a nested compatible digraph
associated with each vertex $v$.

These digraphs associated with root $v_{j}^{m}$ are recursively created
in iteration with each vertex $v_{i}^{n}$, for $1$ to $m-1$, for
$1\leq i\leq k_{m}$.
A non-empty final result generates an output {\em there are compatible antichains} and, otherwise,
{\em there are no compatible antichains}.
Write a nested digraph associated with each {\bf edge} of this digraph from
the Linearized Digraph. Edges will play the role of
vertices and the arrows connecting the vertices are labeled. We will obtain
labeled digraphs whose labels are the pointers of compatible antichains.

The search for $maxG(v)$ is polynomially bounded, as we show in the analysis of computational boundaries in
Proposition \ref{MAXNSTD:prp}.

\begin{procedure}[Linear Digraphs]\label{LIN:prc}
Given a linearized digraph endowed with $p$ branches,
$Br=\{B_{1},\dots,B_{k}\}$, for all $1< i\leq l$, we will write
a set of nested digraphs for each vertex $e_{j}^{t}$, $1< t\leq k$ and
$1\leq j\leq l_{t}$.
The $k^{th}$ step generates either: \\Empty digraphs and an output,
{\bf there are no compatible antichains} or\\
Non empty digraphs  and an output,
{\bf there are compatible antichains}.

{\bf Step 1:} For all $e^{1}_{i}$ in the first branch, $Br_{1}$, for
all $e^{r}_{j}$, $2\leq r\leq l$, if $e^{1}_{i}$ and $e_{r}^{j}$ are compatible, write the digraph,
\[
G(e^{r}_{j},e^{1}_{i})=\langle\{e^{1}_{i},e^{r}_{j}\},\{e^{r}_{j}\rightarrow e^{1}_{i}\},
(label^{1})^{r}_{j}(e^{r}_{j}\rightarrow e^{1}_{j})= \{e^{1}_{i}\}\rangle
\]
Define $G(e^{r}_{j})$ as the union of all $G(e^{r}_{j},e^{1i})$, $1\leq i\leq k_{1}$.

All vertices in the columns $i$, $2\leq i\leq k$ are form an edge with the
compatible vertices in the first column.

{\bf Step $\mathbf{r}$:} Suppose we performed all steps from 1 to $r-1$ and
obtained nested digraphs.

Let
$G(v^{s}_{j})=\langle V_{v^{s}_{j}},E_{v^{s}_{j}},label_{v^{s}_{j}}\rangle$
be nested rooted digraph of depth $r-1$, associated to the
vertices $v^{s}_{j}$, $r\leq s\leq k$, $1\leq j\leq l_{s}$. 

For all $G(v^{r}_{i})$, $1\leq i\leq l_{r}$, for all $v^{s}_{j}$, $s>r$,
$1\leq j\leq k_{s}$, if $v^{r}_{i}$, and $v^{s}_{j}$ are
compatible, replace the (root) vertex $v^{r}_{i}$ in $G(v^{r}_{i})$ by $v^{s}_{j}$.
Obtain the digraph $G(v^{r}_{i}/v^{s}_{j})$.
\begin{enumerate}
\item Let $G(v^{s}_{j},v^{r}_{i})'''$ be the intersection of
  $G(v^{r}_{i}/v^{s}_{j})$ and $G(v^{s}_{j})$;
\item Let $maxG(v^{s}_{j},v^{r}_{i})=
  \langle (V''^{s}_{j})^{r}_{i},(E''^{s}_{j})^{r}_{i},
  (label'^{s}_{j})^{r}_{i}\rangle$ be the maximal nested digraph contained in
  $G(v^{s}_{j},v^{r}_{i})'''$;
\item Define $G(v^{s}_{j},v^{r}_{i})'=\langle (V'^{s}_{j})^{r}_{i},(E'^{s}_{j})^{r}_{i},
  label_{0}^{sjri}\rangle$ as
\[
\begin{array}{ll}
(V'^{s}_{j})^{r}_{i}=  & (V''^{s}_{j})^{r}_{i}\cup\{v^{r}_{i}\}\\
(E'^{s}_{j})^{r}_{i}=  &  (E''^{s}_{j})^{r}_{i}\cup\{v\leftarrow v^{r}_{i}|
v~{\rm belongs~to~level}~r-1\}\\
  (label^{s}_{j})^{r}_{i}(e)=  & label'v^{s}_{j}(e)\cup\{v^{r}_{i}\}\\
  (label^{s}_{j})^{r}_{i}(v^{r}_{i}\rightarrow v)=  & \{v^{r}_{i}\}
\end{array}
\]
\end{enumerate}

Define $G(v^{s}_{j})$ as the union of all $G(v^{s}_{j},v^{r}_{i})'$.

Call the {\em set of nested antichains}\index{set of nested antichains}
to the set of digraphs obtained at each vertex of level $p$.
\end{procedure}

\begin{theorem}\label{ANTICHAIN:thm}
  For all vertex $v$, the digraph $G(v)$, \ref{LIN:prc}, is a
  compatible nested digraph.
\end{theorem}
\begin{proof}
  In the first step, the sequences have the form
  $e\rightarrow v^{1}_{i}$, if both edges are compatible. The label associated to
  the sequences $e\rightarrow v^{1}_{i}$ is $v^{1}_{i}$. The union of such
  digraphs is
  a compatible nested digraph.

  At the step $r$, we replace the vertex $G(v_{ri})$ by $v_{sj}$,
  intersect $G(v^{r}_{i}/v^{s}_{j})$ and $G(v^{s}_{j})$. By construction,
  $max(G(v^{r}_{i}/v^{s}_{j})\cap G(v^{s}_{j}))$ is the biggest compatible
  nested digraph contained in $G(v^{r}_{i}/v^{s}_{j})\cap G(v^{s}_{j})$.

  The digraph
  $G(v^{s}_{j},v^{r}_{i})'$ was obtained by adding the vertex $v^{r}_{i}$,
  compatible
  with all vertices in $G(v^{r}_{i}/v^{s}_{j})\cap G(v^{s}_{j})$.

  Lastly, the union all $G(v^{s}_{j},v^{r}_{i})$ is compatible because it is the
  union of compatible digraphs.
\end{proof}

Finally, highlight that we reached our claim. We show that the nested
digraphs, if non-empty, contain all compatible antichains.

\begin{definition}\label{MAXNESTPTH:dfn}
  Call a {\em maximal nested path over a nested digraph}\index{maximal
  nested path over a nested digraph} to any maximal path
  \[
   v^{k}\rightarrow v^{k-1}\rightarrow\dots \rightarrow v^{2}\rightarrow v
  \]
  so that for all $1\leq i\leq k-2$,
  $label(v^{i},v^{i+1})\subseteq label(v^{i+1},v^{i+2})$.
\end{definition}

Identify the set of all maximal nested paths over a nested digraph with the set
of maximal compatible antichains. Conclude,

\begin{theorem}\label{ANTICHAINLIN:thm}
  A pivoted \thsat\ $\Psi$ is unsatisfiable if and only if the $k^{th}$
  iteration of nested antichains is empty.
\end{theorem}
\section{A Pivoted Strong Version of a \thsat\ Formula}\label{THPVTH}

Given a \thsat\ formula $\Psi$, we show that there are at least one formula $\Psi_{T}$ so that if $\Psi$ is unsatisfiable, so is $\Psi_{T}$.

Arbitrarily choose pairs of literals
$S^{0}=\{p_{01},\neg p_{01},\dots,p_{0k},\neg p_{0k}\}$ in $\Psi$ and factorize $\Psi$,
\[
  (p_{01}\vee S_{p_{01}}) \wedge (\neg p_{01}\vee S_{\neg p_{01}}) 
  \wedge \dots \wedge
  (p_{0k}\vee S_{p_{0k}}) \wedge (\neg p_{0k}\vee S_{\neg p_{0k}}) \wedge
  S_{3}^{1}
  \]
  so that none of the literals of $S^{0}$ appears in any formula in the set of
  formulas
  $\{S_{p_{01}},S_{\neg p_{01}},\dots,S_{p_{0k}},S_{\neg p_{0k}},S_{3}^{1}\}$.

Successively, do the partitions:

\[
\begin{array}{ll}
S_{3}^{1}\equiv\!\!\! &\!\!\! (p_{11}\vee S_{p_{11}})\wedge
  (\neg p_{11}\vee S_{\neg p_{11}}) 
  \wedge \dots \wedge (p_{1k_{1}}\vee S_{p_{1k_{1}}}) \wedge 
  (\neg p_{1k_{1}}\vee S_{\neg p_{1k_{1}}}) \\
  & \wedge S_{3}^{2}\\
S_{3}^{2}\equiv\!\!\! &\!\!\! (p_{21}\vee S_{p_{21}}) \wedge 
  (\neg p_{21}\vee S_{\neg p_{21}}) \wedge \dots \wedge
  (p_{2k_{2}}\vee S_{p_{2k_{2}}}) \wedge 
  (\neg p_{2k_{2}}\vee S_{\neg p_{2k_{2}}}) \\
  & \wedge S_{3}^{3}\\
\dots\\
S_{3}^{h}\!\! \equiv &\!\!\! (p_{h1}\vee S_{p_{h1}}) \wedge 
  (\neg p_{h1}\vee S_{\neg p_{h1}}) \wedge \!\!\dots\!\! \wedge
  (p_{hk_{h}}\vee S_{p_{hk_{h}}}) \wedge 
  (\neg p_{hk_{h}}\vee S_{\neg p_{hk_{h}}})\\
  & \wedge S_{\top} \\
\end{array}
\]
where no literal in the set of conjugated pairs
$\{p_{r1},\neg p_{r1},\dots,p_{rk_{r}},\neg p_{rk_{r}}\}$ belongs to the set of literals of $S_{3}^{s}$, for any $s>r$
or the set of literals of any
$\cup\{S_{p_{rt}}\cup S_{\neg p_{rt}}|1\leq r\leq k_{r}\}$. Moreover
$S_{\top}$ is a $\twsat$ formula.

Finally, obtain the below partition:
\[
  \begin{array}{c}
    \begin{array}{cccc}
    \rnode{11}{p_{01}} & \rnode{12}{\psframebox{S_{p_{01}}}} & 
    \rnode{13}{\neg p_{01}} & \rnode{14}{\psframebox{S_{\neg p_{01}}}}\\[.1cm] 
    \rnode{21}{p_{02}} & \rnode{22}{\psframebox{S_{p_{02}}}} & 
    \rnode{23}{\neg p_{02}} & \rnode{24}{\psframebox{S_{\neg p_{02}}}}\\[.5cm]
    \rnode{31}{p_{0k}} & \rnode{32}{\psframebox{S_{p_{0k}}}} & 
       \rnode{33}{\neg p_{0k}} & \rnode{34}{\psframebox{S_{\neg p_{0k}}}}\\
    \ncline[nodesep=3pt,linestyle=dotted]{-}{22}{32}
    \ncline[nodesep=3pt,linestyle=dotted]{-}{24}{34}
    \end{array}
    \\[-.5cm]
    \begin{array}{c}
    \rnode{4c}{\psframebox{
  \begin{array}{c}
    \begin{array}{cccc}
    \rnode{11}{p_{11}} & \rnode{12}{\psframebox{S_{p_{11}}}} & 
    \rnode{13}{\neg p_{11}} & \rnode{14}{\psframebox{S_{\neg p_{11}}}}\\[.1cm]
    \rnode{21}{p_{21}} & \rnode{22}{\psframebox{S_{p_{21}}}} & 
    \rnode{23}{\neg p_{21}} & \rnode{24}{\psframebox{S_{\neg p_{21}}}}\\ [.5cm]
    \rnode{31}{p_{1k_{1}}} & \rnode{32}{\psframebox{S_{p_{1k_{1}}}}} & 
    \rnode{33}{\neg p_{1k_{1}}} &\rnode{34}{\psframebox{S_{\neg p_{1k_{1}}}}}\\
    \ncline[nodesep=3pt,linestyle=dotted]{-}{22}{32}
    \ncline[nodesep=3pt,linestyle=dotted]{-}{24}{34}
    \end{array}
    \\[-.5cm]
    \vdots \\
    \begin{array}{c}
    \rnode{5c}{\psframebox{
\begin{array}{c}    \begin{array}{cccc}
    \rnode{11}{p_{h1}} & \rnode{12}{\psframebox{S_{p_{h1}}}} & 
    \rnode{13}{\neg p_{h1}} & \rnode{14}{\psframebox{S_{\neg p_{h1}}}}\\[.1cm] 
    \rnode{21}{p_{h2}} & \rnode{22}{\psframebox{S_{p_{h2}}}} & 
    \rnode{23}{\neg p_{h2}} & \rnode{24}{\psframebox{S_{\neg p_{h2}}}}\\[.5cm] 
    \rnode{31}{p_{hk_{h}}} & \rnode{32}{\psframebox{S_{p_{hk_{h}}}}} & 
  \rnode{33}{\neg p_{hk_{h}}} & \rnode{34}{\psframebox{S_{\neg p_{hk_{h}}}}}\\
    \ncline[nodesep=3pt,linestyle=dotted]{-}{22}{32}
    \ncline[nodesep=3pt,linestyle=dotted]{-}{24}{34}
    \end{array}\\[-.5cm]
S_{\top}
\end{array}
                          }} 
    \end{array}
  \end{array}
                          }} 
    \end{array}
  \end{array}
  \]

We show that there is a modified pivoted \thsat, $\Psi_{T}$ so that $\Psi_{T}$ is unsatisfiable if and only if $\Psi$ is unsatisfiable. 

Let
\[
\{r_{21},\neg r_{21},\!\dots,
r_{1k_{1}},\neg r_{1k_{1}},\!\dots,r_{h1},\neg r_{h1},\!\dots,r_{hk_{h}},\neg r_{h_{h}}\}
\]
be a set of literals disjoint from the set $\letter(\Psi)$. 
Let $\Psi_{T}$ be,
\[
\begin{array}{l}
  (p_{11}\vee S_{p_{11}})\wedge(\neg p_{11}\vee S_{\neg p_{11}})\wedge\dots\wedge
  (p_{1k_{1}}\vee S_{p_{1k_{1}}})\wedge(\neg p_{1k_{1}}\vee S_{\neg p_{1k_{1}}})\wedge\\
 (r_{21}\vee (S_{p_{21}}\wedge\neg p_{21}))\wedge
 (\neg r_{21}\vee (S_{\neg p_{21}}\wedge p_{21}))\wedge\dots\wedge\\
(r_{2k_{2}}\vee (S_{p_{2k_{2}}}\wedge\neg p_{2k_{2}}))\wedge
(\neg r_{2k_{2}}\vee (S_{\neg p_{2k_{2}}}\wedge p_{2k_{2}}))\wedge\dots\wedge\\
(r_{h1}\vee (S_{p_{h1}}\wedge\neg p_{h1}))\wedge
(\neg r_{h1}\vee (S_{\neg p_{h1}}\wedge p_{h1}))\wedge\dots\wedge\\
(r_{hk_{h}}\vee (S_{p_{h1}}\wedge\neg p_{h1}))\wedge
(\neg r_{hk_{h}}\vee (S_{\neg p_{h1}}\wedge p_{h1}))\wedge
 (t\vee S_{\top})\wedge(\neg t\vee S_{\top})
  \end{array}
\]

Write $\Psi$ as its factorized version,
\[
\begin{array}{l}
  (p_{11}\vee S_{p_{11}}) \wedge (\neg p_{11}\vee S_{\neg p_{11}}) 
  \wedge \dots \wedge
  (p_{1k_{1}}\vee S_{p_{1k_{1}}}) \wedge (\neg p_{1k_{1}}\vee S_{\neg p_{1k_{1}}})\\
  (p_{21}\vee S_{p_{21}})\wedge
  (\neg p_{21}\vee S_{\neg p_{21}}) 
  \wedge \dots \wedge (p_{2k_{2}}\vee S_{p_{2k_{2}}}) \wedge 
  (\neg p_{2k_{2}}\vee S_{\neg p_{2k_{2}}}) \\
  \wedge 
  (\neg p_{2k_{2}}\vee S_{\neg p_{2k_{2}}}) \\
  \dots\\
(p_{h1}\vee S_{p_{h1}}) \wedge 
  (\neg p_{h1}\vee S_{\neg p_{h1}}) \wedge \!\!\dots\!\! \wedge
  (p_{hk_{h}}\vee S_{p_{hk_{h}}}) \wedge 
  (\neg p_{hk_{h}}\vee S_{\neg p_{hk_{h}}})\\
  \wedge S_{\top} \\
\end{array}
\]

We show that $\Psi_{T}$ is unsatisfiable if and only if $\Psi$ is unsatisfiable. 

Let $\Sigma$ be the set of all mappings from $\{11,\dots,1k_{1},\dots,h1,\dots,hk_{h}\}$ onto $\{True,False\}$.
Note that $\Sigma$ has $2^{k_{1}+\dots+k_{h}}$ elements.

Let $\sigma\in\Sigma$. For all conjugated pair $\{p_{uv},\neg p_{uv}\}$,
$1\leq u\leq h$ and $1\leq v\leq k_{v}$, let $\epsilon_{uv}=p_{uv}$, if $\sigma(p_{uv})=True$ and, otherwise, $\epsilon_{uv}=\neg p_{uv}$.
Let $\upsilon_{uv}$ be $r_{uv}$ if $\epsilon_{uv}=\neg p_{uv}$ and, otherwise  $\upsilon_{uv}=\neg r_{uv}$.
Given a Boolean formula $F$, $F/\sigma$ denotes the substitution of all literals in $F$ by its value over $\sigma$.

Using the valuation $\sigma$, we have, respectively, for $\Psi$ and $\Psi_{T}$,
\[
\begin{array}{l}
  S_{\neg \epsilon_{11}}/\sigma\wedge \dots \wedge S_{\neg \epsilon_{1k_{1}}}/\sigma\wedge
  S_{\neg \epsilon_{21}}/\sigma\wedge \dots \wedge S_{\neg \epsilon_{2k_{2}}}/\sigma\wedge 
  \dots\wedge\\
  S_{\neg \epsilon_{h1}}/\sigma \wedge \dots
  \wedge S_{\neg \epsilon_{hk_{h}}}/\sigma
  \wedge S_{\top} 
\end{array}
\]
and
\[
\begin{array}{l}
 S_{\neg \epsilon_{11}}/\sigma\wedge \dots \wedge S_{\neg \epsilon_{1k_{1}}}/\sigma\wedge
 \neg\upsilon_{21}\wedge (\upsilon_{21}\vee S_{\neg\epsilon_{21}}/\sigma)\wedge\dots\wedge\\
  \neg\upsilon_{2k_{2}}\wedge (\upsilon_{2k_{2}}\vee S_{\neg\epsilon_{2k_{2}}}/\sigma)\wedge
 \neg\upsilon_{h1}\wedge(\upsilon_{h1}\vee S_{\neg \epsilon_{h1}}/\sigma)\wedge\dots\wedge\\
 \neg\upsilon_{hk_{h}}\wedge(\upsilon_{hk_{h}}\vee S_{\neg \epsilon_{h1}}/\sigma)
 \wedge S_{\top} 
  \end{array}
\]

Whether at least one $\neg\upsilon_{lm}$ is {\em false} then, $\Psi_{T}$ is false. Otherwise, all $\neg\upsilon_{lm}$ is {\bf true} and $\Psi_{T}$ is
\[
\begin{array}{l}
  S_{\neg \epsilon_{11}}/\sigma\wedge \dots \wedge S_{\neg \epsilon_{1k_{1}}}/\sigma\wedge
  S_{\neg \epsilon_{21}}/\sigma\wedge \dots \wedge S_{\neg \epsilon_{2k_{2}}}/\sigma\wedge 
  \dots\wedge\\
  S_{\neg \epsilon_{h1}}/\sigma \wedge \dots
  \wedge S_{\neg \epsilon_{hk_{h}}}/\sigma
  \wedge S_{\top} 
\end{array}
\]
and, under this conditions of valuation, $\Psi$ is unsatisfiable if and only  if $\Psi_{T}$ is unsatisfiable.
\section{Bounds on Computation}\label{BC:sec}
Open this section with the study of the bounds in Space,
$\mathbf{SPACE}(f(n))$ and in time, $\mathcal{O}(f(n))$. Both bounds
ensure that we do not trade space by time or vice-versa in our
considerations. 

Given a pivoted \thsat\ formula, we follow the polynomially bounded steps,
\begin{enumerate}
\item Write the cylindrical digraph;
\item Write the closed digraphs;
\item Write the linearized digraph;
\item Write the maximal nested digraphs. A special care has to be taken on
  writing the maximal nested digraph contained in the intersection.
\end{enumerate}

We prove  less straightforward bounds in time and space.

\begin{lemma}\label{MAXLINPOLI:lmm}
The cylindrical digraph, $\clndr=\langle V,E,label\rangle$ is written in
polynomial time and space. 
\end{lemma}
\begin{proof}
The size of $\clndr$-graph is given by
\begin{enumerate}
\item The size of $V$, the set of vertices is the size of literals, $L$;
\item The size of $E$ is bounded by the square of the number of literals, $L$
  $|E|\leq |L|^{2}$\index{$|E|\leq |L|^{2}$}. Indeed, the arrows are bounded
  by the number of arrows in a polyhedron with $|L|$ sides.
\end{enumerate}
\end{proof}

There is an algorithm polynomially bounded in time to write a closed digraph.
We describe the, polynomial in time, search for nonempty intervals $[a,\neg a]$. Recall that no loops are allowed. Indeed, if we write a loop, that means that
if an incompatible combination $Inc=ij_{1},\dots, ij_{r}$, any compatible
combination that contains $Inc$ is likewise incompatible.

Divide the search algorithm into two procedures,
  \begin{enumerate}
  \item Write $[a,q[$, the subdigraph of $\clsddg$ of all sequences
    connected to $a$ (no specified end, just source, $a$);
  \item If $[a,q[\neq\emptyset$, from $[a,q[$, write $[a,\neg a]$,
    the subdigraph of paths between $a$ and $\neg a$ contained in $[a,q[$. 
  \end{enumerate}
 \begin{pseudocode}\label{INTRVL:psc}
      {\bf Part I, write an Interval $[a,q[$}
  \begin{algorithmic}          
   \State Input cylindrical digraph $Cyl=\langle V,E,label\rangle$.
   \State $[a,q[=\langle \{a\},\emptyset,\emptyset\rangle=
                 \langle V_{[a,q[},E_{[a,q[},label_{[a,q[}\rangle$
\Comment{Initial State}
\State $E_{aux}=E\setminus E_{[a,q[}$
\State $V_{aux}=V\setminus V_{[a,q[}$
\Comment{No loops}
\While{$E_{aux}\neq\emptyset$}
\State $E_{[a,q[}\gets E_{[a,q[}\cup\{c\Rightarrow b|b\in V_{aux}~AND~
            c\Rightarrow b\in E\} $
\State $V_{[a,q[}\gets V_{[a,q[}\cup\{b|\exists c\in V_{[a,q[}
               (c\Rightarrow b\in E_{aux})\}$
\State $label_{[a,q[}\gets label_{[a,q[}\cup\{label(c\Rightarrow b)|
                    \exists b\in V_{[a,q[} (c\Rightarrow b\in E_{aux}) \}$
\State $E_{aux}=E\setminus E_{[a,q[}$
\State $V_{aux}=V\setminus V_{[a,q[}$
\EndWhile
\end{algorithmic}
 \end{pseudocode}
 
Once we obtain
$[a,q[=\langle V_{[a,q[},E_{[a,q[},label_{[a,q[}\rangle$, if
$\neg a\in V_{[a,q[}$, we perform the above algorithm backwards, that is
we start with
$]q,\neg a]=\langle \{\neg a\},\emptyset,\emptyset\rangle=
\langle V_{]q,\neg a]},E_{]q,\neg a]},label_{]q,\neg a]}\rangle$
and perform Pseudocode \ref{INTRVL:psc} taking care
to search over arrows $b\Rightarrow\neg a$ and obtain $[a,\neg a]$.

\begin{observation}
  As Pseudocode \ref{INTRVL:psc}, at each iteration, updates the set
  $E_{aux}$, we avoid sequences of the form,
  \[
  r\Rightarrow q_{1}\Rightarrow \dots\Rightarrow r
  \] 
  that is, we do not have loops.
\end{observation}

Let us scrutinize the size of the closed digraph.
\begin{lemma}\label{SIZECLSD:lmm}
  The set of closed digraphs is the union of at most $|V|$ digraphs with
  $|V|^{2}$ vertices and $|V|^{3}$ edges.
\end{lemma}
\begin{proof}
  The set of necessarily true pairs has, at most, the size of $|V|$, the number
  of literals. Each interval $[\neg p,p]$ has its size bounded by the size
  of the cylindrical digraph, that is, at most $|V|$ vertices and $|V|^{2}$
  edges, that is, a total of $|V|^{2}$ vertices and $|V|^{3}$ edges.

  Obtain the union of, at most, $|V|$ digraphs with $|V|$ vertices and $|V|^{2}$
  edges, $\mathbf{SPACE}(|V|^{3})$.
\end{proof}

\begin{lemma}\label{TIMECLSD:lmm}
The search for the set of closed digraphs is polynomially in time.
\end{lemma}
\begin{proof}
  We perform at most $|V|$ searches, one to each interval $[p,\neg p]$
  for $p$ necessarily true.

  We search over the number of vertices and obtain,
  \[
  \begin{array}{l}
    |V|-1\\
    |V|-1-r_{1}\\
    \vdots\\
    |V|-1-r_{1}-\dots-r_{q}
  \end{array}
  \]
  where $1,r_{1},\dots,r_{q}$ are the vertices we removed from the vertices
  that form the set of edges $E_{aux}$.

  The maximum under the conditioning $1+r_{1}+\dots+r_{q}=V$, leads to a maximum
  $r_{1}=\dots=r_{q}=V/q$ and the search goes to a top of $|V|^{2}$ searches
  over $|V|$ vertices and the search is bonded by $|V|^{3}$ in time.
\end{proof}

Continue reviewing the effort to Lift, Multiply Branches and
the Addition of Label. After that, we must analyze the whole process of
writing the nested digraphs associated with each edge.

\begin{proposition}\label{LAM:prp}
  Given a closed digraph, the operations of Lifting, Adding Labels
  and Multiplying Branches are linearly bounded.
\end{proposition}
\begin{proof}
  Lifting together with Multiplication of branches multiply the number of any
  branching $a$ accordingly to the maximum of branches in $Up(a)$ or $Down(a)$.
 
  Adding Label adds several new labels to, at most the number of edges
  in the closed digraph its addition is bounded by the number of branching
  and Multiplication of branches is bounded by the maximum number of
  branches, that is, the maximum number if edges, $|V|^{3}$.
\end{proof}

Our next question is: {\em What the size of the Linearized Digraph is?}

The set of closed digraphs is the union of at most $|V|$ digraphs. Each closed
digraph has the size bounded by $|V|^{2}$ vertices and $|V|^{3}$ edges, by
Lemma \ref{SIZECLSD:lmm}. As we deal with linear digraphs, the number of edges
is bounded by the number of vertices in the linearized digraph.

We consider a maximum multiplication of branches. The maximum of multiplication operations we perform is bounded by $|V|^{2}$ vertices. Suppose, in the very pessimistic valuation  we have $|V|^{2}$ columns with $|V|^{2}$ edges.

The edges are bounded by $|V|^{2}$ because we have linear digraphs and the vertices do not form branching.

\begin{lemma}\label{SIZENESTED:lmm}
  Given a Linearized digraph endowed with $|V|^{2}$ columns with $|V|^{2}$ edges.
  Then, the process of building the nested digraphs has the following bounds,
\end{lemma}
\begin{proof}
  {\bf First Step:} In the first step, each vertex in the first column
  interacts with each vertex on the remainder columns, that is, we built at most
  $|V|\times(|V|\times(|V|-1))$ nested digraphs, Each rooted digraph, has a
  maximum of $|V|$ edges because all vertices are connected to a single root.
  Henceforth, we use a maximum of $V|^{2}$ edges to keep optimization simpler.
  
  {\bf $\mathbf{r^{th}}$ Step:} We have at most $|V|$ interactions over
  $|V|\times(|V|-r-1)$ vertices that form the root associated to the vertex
  of each column. That is, we generate $|V|\times(|V|-r-1)$ rooted digraphs.
  Each digraph has depth $r$ and each column has a maximum of $|V|$
  vertices. We have a total bounded by a maximum of $r\times |V|$ vertices
  and $r\times |V|^{2}$ edges.

  {\bf ${\mathbf{|V|^{th}}}$ Step:} In the $|V|-1$ step, obtain a maximum of
  $|V|$ digraphs with depth $|V|$. We have a top of $|V|$ vertices in
  each column, that is, $|V|^{2}$ vertices and $|V|^{3}$ edges.
\end{proof}

\begin{proposition}\label{MAXNSTD:prp}
  Given two nested rooted digraphs, whose root is $v$, endowed with
  $R\times T$ vertices, located in $T$ columns, $R$ vertices at each column.
  $R^{2}\times T$ edges, then the operation of writing the maximal nested
  digraph contained in their intersection is polynomially bounded.
\end{proposition}
\begin{proof}
  The first algorithm shows the union of linear digraphs whose
  labels are ordered by $\subseteq$.

  From $m$ ranging from levels 1 and 2 until $T-2$ and $T-1$, search over
  a maximum of $S\times T$ edges. At each edge, we perform searches over
  $S$ edges whose vertices lie in levels $m+1,m+2$. Obtain at most
  $S^{2}$ interactions over levels $m+1,m+2$. At the end of the process,
  obtain a total of $S^{2}\times T$ searches.
  
  In the second pseudocode, for edges in the levels $m=T$ and $T-1$,
  step $-1$ until levels $m=2$ and $1$, write, for any edge
  $u^{m}_{i}\rightarrow v^{m-1}_{j}$, the subdigraph $G_{u^{m}_{i},v^{m-1}_{j}}$
  that contains all maximal sequences that start at level $k$, pass through
  $u^{m}_{i}\rightarrow v^{m-1}_{j}$ and end in $v$. Rule out from the set of
  labels any edge in the level $m$, except $u^{m}_{i}$ and in the level $m-1$,
  except $v^{m-1}_{j}$.

  Each edge visits, in a maximum of $|V|^{2}$ edges. We have, in the worst
  case, are bounded by the number of closed digraphs, and the use of dimension
  $|V|^{3}$. so, we search  $|V|^{5}$ in time.  
\end{proof}
\section{Examples}\label{EX:sec}
\begin{example}\label{EX1:exp}
Consider the $\clsddg$,
\[
\begin{array}{c@{\hskip 80pt}c}
\rnode{11}{p_{1}} & \rnode{12}{p_{2}}\\[.4cm] 
\rnode{21}{p_{3}} & \rnode{22}{p_{4}}\\[.4cm] 
\rnode{31}{p_{5}} & \rnode{32}{p_{6}}\\[.4cm] 
\rnode{41}{p_{7}} & \rnode{42}{p_{8}}
\ncline[doubleline=true,nodesep=5pt]{->}{11}{21}
\ncline[doubleline=true,nodesep=5pt]{->}{11}{22}
\ncline[doubleline=true,nodesep=5pt]{->}{12}{21}
\ncline[doubleline=true,nodesep=5pt]{->}{12}{22}
\ncline[doubleline=true,nodesep=5pt]{->}{21}{31}
\ncline[doubleline=true,nodesep=5pt]{->}{22}{32}
\ncline[doubleline=true,nodesep=5pt]{->}{21}{32}
\ncline[doubleline=true,nodesep=5pt]{->}{22}{31}
\ncline[doubleline=true,nodesep=5pt]{->}{31}{41}
\ncline[doubleline=true,nodesep=5pt]{->}{31}{42}
\ncline[doubleline=true,nodesep=5pt]{->}{32}{41}
\ncline[doubleline=true,nodesep=5pt]{->}{32}{42}
\end{array}
\]

We do not focus on the compatibility among vertices. It is not relevant for
this analysis.

The task of writing the chains is clearly expsize, as we show below\\
\begin{tabular}{llllllllllllllllllllllll
  lllllllllllllllllllllll}
 $p_{1}$ & $p_{1}$ & $p_{1}$ & $p_{1}$ & $p_{1}$ & $p_{1}$ & $p_{1}$ & $p_{1}$ &
 $p_{2}$ & $p_{2}$ & $p_{2}$ & $p_{2}$ &  $p_{2}$ & $p_{2}$ & $p_{2}$ & $p_{2}$ \\
$p_{3}$ & $p_{3}$ & $p_{3}$ & $p_{3}$ & $p_{4}$ & $p_{4}$ & $p_{4}$&$p_{4}$ &
$p_{3}$ & $p_{3}$ & $p_{3}$ & $p_{3}$ & $p_{4}$ & $p_{4}$ & $p_{4}$&$p_{4}$\\
$p_{5}$ & $p_{5}$ & $p_{6}$ & $p_{6}$ & $p_{5}$ &$p_{5}$ & $p_{6}$ &$p_{6}$&
 $p_{5}$ &$p_{5}$ &$p_{6}$&$p_{6}$&  $p_{5}$ &  $p_{5}$& $p_{6}$& $p_{6}$ \\
$p_{7}$ & $p_{8}$ & $p_{7}$ & $p_{8}$ & $p_{7}$ &$p_{8}$ & $p_{7}$ &$p_{8}$&
 $p_{7}$ &$p_{8}$ &$p_{7}$&$p_{8}$&  $p_{7}$ &  $p_{8}$& $p_{7}$& $p_{8}$ 
\end{tabular}

Solve using our method. Multiply roots, 
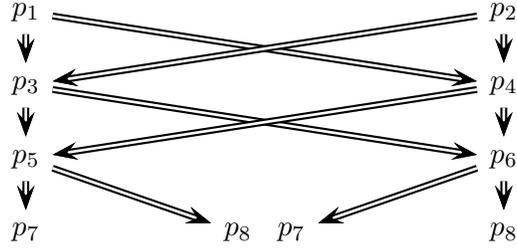
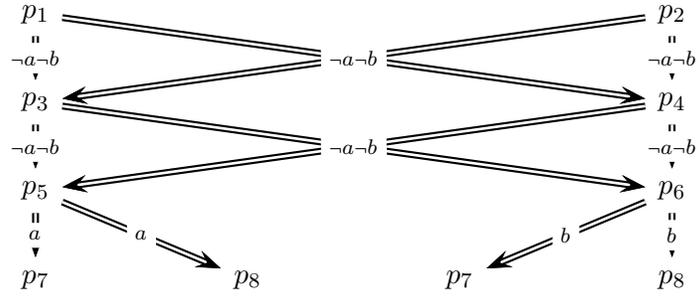
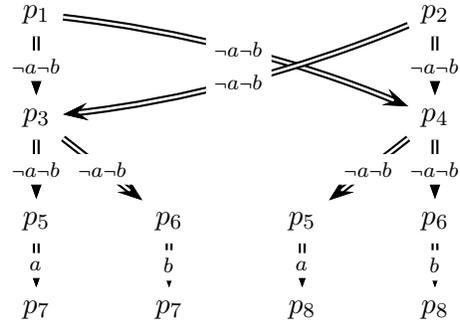
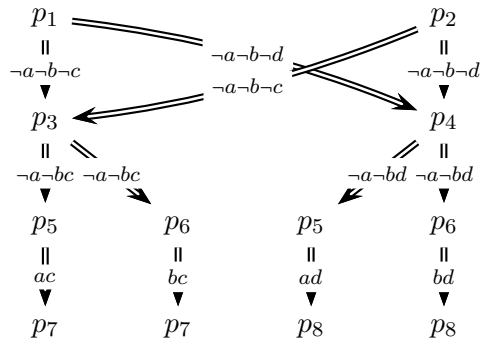
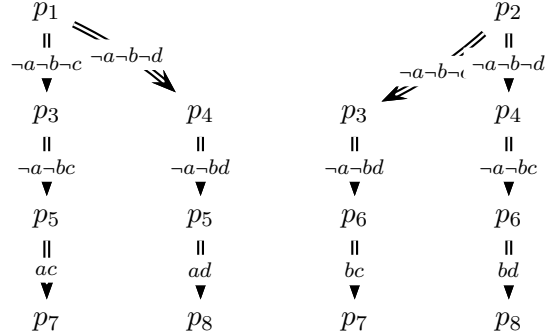
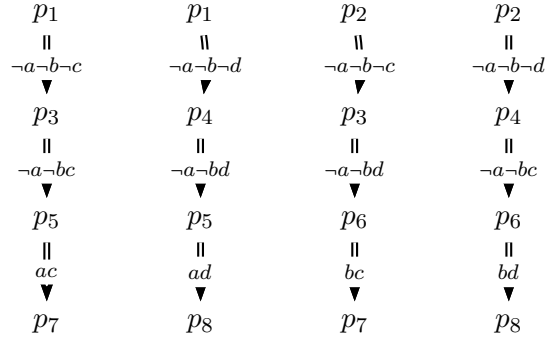
\[
\begin{array}{c@{\hskip 70pt}cc@{\hskip 70pt}c}
\rnode{11}{p_{1}} &&&   \rnode{12}{p_{2}}\\[.4cm] 
\rnode{21}{p_{3}} &&&   \rnode{22}{p_{4}}\\[.4cm] 
\rnode{31}{p_{5}} & &&  \rnode{32}{p_{6}}\\[.4cm] 
\rnode{411}{p_{7}} &\rnode{421}{p_{8}}&\rnode{412}{p_{7}}& \rnode{422}{p_{8}}
\ncline[doubleline=true,nodesep=5pt]{->}{11}{21}
\ncline[doubleline=true,nodesep=5pt]{->}{11}{22}
\ncline[doubleline=true,nodesep=5pt]{->}{12}{21}
\ncline[doubleline=true,nodesep=5pt]{->}{12}{22}
\ncline[doubleline=true,nodesep=5pt]{->}{21}{31}
\ncline[doubleline=true,nodesep=5pt]{->}{22}{32}
\ncline[doubleline=true,nodesep=5pt]{->}{21}{32}
\ncline[doubleline=true,nodesep=5pt]{->}{22}{31}
\ncline[doubleline=true,nodesep=5pt]{->}{31}{411}
\ncline[doubleline=true,nodesep=5pt]{->}{31}{421}
\ncline[doubleline=true,nodesep=5pt]{->}{32}{412}
\ncline[doubleline=true,nodesep=5pt]{->}{32}{422}
\end{array}
\]

Add labels, obtain, 
\[
\begin{array}{c@{\hskip 70pt}c@{\hskip 70pt}c@{\hskip 70pt}c}
\rnode{11}{p_{1}} &&& \rnode{12}{p_{2}}\\[.6cm] 
\rnode{21}{p_{3}} &&& \rnode{22}{p_{4}}\\[.5cm] 
\rnode{31}{p_{5}} &&& \rnode{32}{p_{6}}\\[.5cm] 
\rnode{411}{p_{7}} &\rnode{421}{p_{8}}&\rnode{412}{p_{7}}& \rnode{422}{p_{8}}
\ncarc[doubleline=true,nodesep=5pt]{->}{11}{21}\mput*{_{\neg a\neg b}}
\ncarc[doubleline=true,nodesep=5pt]{->}{11}{22}\mput*{_{\neg a\neg b}}
\ncline[doubleline=true,nodesep=5pt]{->}{12}{21}\mput*{_{\neg a\neg b}}
\ncline[doubleline=true,nodesep=5pt]{->}{12}{22}\mput*{_{\neg a\neg b}}
\ncline[doubleline=true,nodesep=5pt]{->}{21}{31}\mput*{_{\neg a}}
\ncline[doubleline=true,nodesep=5pt]{->}{22}{32}\mput*{_{\neg b}}
\ncarc[doubleline=true,nodesep=5pt]{->}{21}{32}\mput*{_{\neg b}}
\ncarc[doubleline=true,nodesep=5pt]{->}{22}{31}\mput*{_{\neg a}}
\ncline[doubleline=true,nodesep=5pt]{->}{31}{411}\mput*{_{a}}
\ncline[doubleline=true,nodesep=5pt]{->}{31}{421}\mput*{_{a}}
\ncline[doubleline=true,nodesep=5pt]{->}{32}{412}\mput*{_{b}}
\ncline[doubleline=true,nodesep=5pt]{->}{32}{422}\mput*{_{b}}
\end{array}
\]
Lifting,
\[
\begin{array}{c@{\hskip 70pt}c@{\hskip 70pt}c@{\hskip 70pt}c}
\rnode{11}{p_{1}} &&& \rnode{12}{p_{2}}\\[.6cm] 
\rnode{21}{p_{3}} &&& \rnode{22}{p_{4}}\\[.5cm] 
\rnode{311}{p_{5}} & \rnode{322}{p_{6}}&\rnode{312}{p_{5}}
    &\rnode{321}{p_{6}} \\[.6cm] 
\rnode{411}{p_{7}}& \rnode{422}{p_{8}} &\rnode{421}{p_{8}}&\rnode{412}{p_{7}}
\ncarc[doubleline=true,nodesep=5pt]{->}{11}{21}\mput*{_{\neg a\neg b}}
\ncarc[doubleline=true,nodesep=5pt]{->}{11}{22}\mput*{_{\neg a\neg b}}
\ncline[doubleline=true,nodesep=5pt]{->}{12}{21}\mput*{_{\neg a\neg b}}
\ncline[doubleline=true,nodesep=5pt]{->}{12}{22}\mput*{_{\neg a\neg b}}
\ncline[doubleline=true,nodesep=5pt]{->}{21}{311}\mput*{_{\neg a}}
\ncline[doubleline=true,nodesep=5pt]{->}{22}{321}\mput*{_{\neg b}}
\ncarc[doubleline=true,nodesep=5pt]{->}{21}{322}\mput*{_{\neg b}}
\ncarc[doubleline=true,nodesep=5pt]{->}{22}{312}\mput*{_{\neg a}}
\ncline[doubleline=true,nodesep=5pt]{->}{311}{411}\mput*{_{a}}
\ncline[doubleline=true,nodesep=5pt]{->}{312}{421}\mput*{_{a}}
\ncline[doubleline=true,nodesep=5pt]{->}{321}{412}\mput*{_{b}}
\ncline[doubleline=true,nodesep=5pt]{->}{322}{422}\mput*{_{b}}
\end{array}
\]
Add labels, obtain, 
\[
\begin{array}{c@{\hskip 70pt}c@{\hskip 70pt}c@{\hskip 70pt}c}
\rnode{11}{p_{1}} &&& \rnode{12}{p_{2}}\\[.7cm] 
\rnode{21}{p_{3}} &&& \rnode{22}{p_{4}}\\[.5cm] 
\rnode{311}{p_{5}} & \rnode{322}{p_{6}}&\rnode{312}{p_{5}}
    &\rnode{321}{p_{6}} \\[.7cm] 
\rnode{411}{p_{7}}& \rnode{422}{p_{8}} &\rnode{421}{p_{8}}&\rnode{412}{p_{7}}
\ncarc[doubleline=true,nodesep=5pt]{->}{11}{21}\mput*{_{\neg a\neg b\neg c}}
\ncarc[doubleline=true,nodesep=5pt]{->}{11}{22}\mput*{_{\neg a\neg b\neg d}}
\ncline[doubleline=true,nodesep=5pt]{->}{12}{21}\mput*{_{\neg a\neg b\neg c}}
\ncline[doubleline=true,nodesep=5pt]{->}{12}{22}\mput*{_{\neg a\neg b\neg d}}
\ncline[doubleline=true,nodesep=5pt]{->}{21}{311}\mput*{_{\neg a c}}
\ncline[doubleline=true,nodesep=5pt]{->}{22}{321}\mput*{_{\neg bd}}
\ncarc[doubleline=true,nodesep=5pt]{->}{21}{322}\mput*{_{\neg bc}}
\ncarc[doubleline=true,nodesep=5pt]{->}{22}{312}\mput*{_{\neg ad}}
\ncline[doubleline=true,nodesep=5pt]{->}{311}{411}\mput*{_{ac}}
\ncline[doubleline=true,nodesep=5pt]{->}{312}{421}\mput*{_{ad}}
\ncline[doubleline=true,nodesep=5pt]{->}{321}{412}\mput*{_{bd}}
\ncline[doubleline=true,nodesep=5pt]{->}{322}{422}\mput*{_{bc}}
\end{array}
\]
Lifting,
\[
\begin{array}{c@{\hskip 70pt}c@{\hskip 70pt}c@{\hskip 70pt}c}
\rnode{11}{p_{1}} &&& \rnode{12}{p_{2}}\\[.7cm] 
\rnode{211}{p_{3}} &
     \rnode{221}{p_{4}}&\rnode{212}{p_{3}} & \rnode{222}{p_{4}}\\[.5cm] 
\rnode{311}{p_{5}} & \rnode{321}{p_{6}}&\rnode{322}{p_{6}}&\rnode{312}{p_{5}}\\[.7cm] 
\rnode{411}{p_{7}}& \rnode{412}{p_{7}}&\rnode{422}{p_{8}}& \rnode{421}{p_{8}}
\ncarc[doubleline=true,nodesep=5pt]{->}{11}{211}\mput*{_{\neg a\neg b\neg c}}
\ncarc[doubleline=true,nodesep=5pt]{->}{11}{221}\mput*{_{\neg a\neg b\neg d}}
\ncline[doubleline=true,nodesep=5pt]{->}{12}{212}\mput*{_{\neg a\neg b\neg c}}
\ncline[doubleline=true,nodesep=5pt]{->}{12}{222}\mput*{_{\neg a\neg b\neg d}}
\ncline[doubleline=true,nodesep=5pt]{->}{211}{311}\mput*{_{\neg a c}}
\ncline[doubleline=true,nodesep=5pt]{->}{221}{321}\mput*{_{\neg bd}}
\ncarc[doubleline=true,nodesep=5pt]{->}{212}{322}\mput*{_{\neg bc}}
\ncarc[doubleline=true,nodesep=5pt]{->}{222}{312}\mput*{_{\neg ad}}
\ncline[doubleline=true,nodesep=5pt]{->}{311}{411}\mput*{_{ac}}
\ncline[doubleline=true,nodesep=5pt]{->}{312}{421}\mput*{_{ad}}
\ncline[doubleline=true,nodesep=5pt]{->}{321}{412}\mput*{_{bd}}
\ncline[doubleline=true,nodesep=5pt]{->}{322}{422}\mput*{_{bc}}
\end{array}
\]
Lifting,
\[
\begin{array}{c@{\hskip 70pt}c@{\hskip 70pt}c@{\hskip 70pt}c}
  \rnode{111}{p_{1}} &\rnode{112}{p_{1}}&\rnode{121}{p_{2}}&
     \rnode{122}{p_{2}}\\[.7cm] 
\rnode{211}{p_{3}} &
     \rnode{221}{p_{4}}&\rnode{212}{p_{3}} & \rnode{222}{p_{4}}\\[.5cm] 
\rnode{311}{p_{5}} & \rnode{321}{p_{6}}&\rnode{322}{p_{6}}&\rnode{312}{p_{5}}\\[.7cm] 
\rnode{411}{p_{7}}& \rnode{412}{p_{7}}&\rnode{422}{p_{8}}& \rnode{421}{p_{8}}
\ncarc[doubleline=true,nodesep=5pt]{->}{111}{211}\mput*{_{\neg a\neg b\neg c}}
\ncarc[doubleline=true,nodesep=5pt]{->}{112}{221}\mput*{_{\neg a\neg b\neg d}}
\ncline[doubleline=true,nodesep=5pt]{->}{121}{212}\mput*{_{\neg a\neg b\neg c}}
\ncline[doubleline=true,nodesep=5pt]{->}{122}{222}\mput*{_{\neg a\neg b\neg d}}
\ncline[doubleline=true,nodesep=5pt]{->}{211}{311}\mput*{_{\neg a c}}
\ncline[doubleline=true,nodesep=5pt]{->}{221}{321}\mput*{_{\neg bd}}
\ncarc[doubleline=true,nodesep=5pt]{->}{212}{322}\mput*{_{\neg bc}}
\ncarc[doubleline=true,nodesep=5pt]{->}{222}{312}\mput*{_{\neg ad}}
\ncline[doubleline=true,nodesep=5pt]{->}{311}{411}\mput*{_{ac}}
\ncline[doubleline=true,nodesep=5pt]{->}{312}{421}\mput*{_{ad}}
\ncline[doubleline=true,nodesep=5pt]{->}{321}{412}\mput*{_{bd}}
\ncline[doubleline=true,nodesep=5pt]{->}{322}{422}\mput*{_{bc}}
\end{array}
\]

The only possible compatible combinations, according to compatibility of
labels we added, not considering compatibility among the original label, is,
\[
\begin{array}{l}
  p_{5}\Rightarrow p_{7},p_{6}\Rightarrow p_{7}, p_{5}\Rightarrow p_{8},
         p_{6}\Rightarrow p_{8}\\
  p_{5}\Rightarrow p_{7}, p_{4}\Rightarrow p_{6},
        p_{3}\Rightarrow p_{6}, p_{5}\Rightarrow p_{8}\\
  p_{3}\Rightarrow p_{5}, p_{6}\Rightarrow p_{7}, p_{6}\Rightarrow p_{8},
         p_{4}\Rightarrow p_{5}\\
 p_{3}\Rightarrow p_{5}, p_{4}\Rightarrow p_{6},
        p_{3}\Rightarrow p_{6}, p_{4}\Rightarrow p_{5}\\
 p_{3}\Rightarrow p_{5},p_{1}\Rightarrow p_{4},
        p_{2}\Rightarrow p_{3}, p_{4}\Rightarrow p_{5}\\
 p_{1}\Rightarrow p_{3}, p_{4}\Rightarrow p_{5}, p_{3}\Rightarrow p_{6},
         p_{2}\Rightarrow p_{4}\\
  p_{1}\Rightarrow p_{3}, p_{1}\Rightarrow p_{4},
        p_{2}\Rightarrow p_{3}, p_{2}\Rightarrow p_{4}
\end{array}
\]
\end{example}
\begin{example}\label{EX2:exp}
Consider the two pivoted formulas,
\[
\begin{array}{ll}
  \Psi_{1}=&a_{1}\vee\big( (p_{1}\vee q_{1})\wedge(p_{1}\vee \neg q_{1})\big)\wedge
     \neg a_{1}\vee\big( (p_{1}\vee q_{1})\wedge(p_{1}\vee \neg q_{1})\big)\wedge\\
 &a_{2}\vee\big( (\neg p_{1}\vee q_{2})\wedge(\neg p_{1}\vee \neg q_{2})\big)\wedge
       \neg a_{2}\vee\big( (p_{1}\vee q_{2})\wedge(p_{1}\vee \neg q_{2})\big)\\
  \Psi_{2}=&a_{1}\vee\big( (p_{1}\vee q_{1})\wedge(\neg p_{1}\vee q_{2})\big)\wedge
  \neg a_{1}\vee\big( (p_{1}\vee q_{1})\wedge(\neg p_{1}\vee q_{2})\big)\wedge\\
 &a_{2}\vee\big( (p_{1}\vee \neg q_{1})\wedge(\neg p_{1}\vee\neg  q_{2})\big)\wedge
       \neg a_{2}\vee\big( (p_{1}\vee \neg q_{1})\wedge(\neg p_{1}\vee q_{2})\big)
\end{array}
\]
\end{example}

The formulas $\Psi_{1}$ and $\Psi_{2}$ are, respectively, unsatisfiable and
satisfiable.

Write a fragment of the closed digraph,
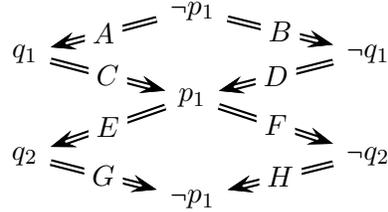
\[
\begin{array}{c@{\hskip 50pt}c@{\hskip 50pt}c}
&\rnode{11}{\neg p_{1}} &\\[.1cm] 
\rnode{21}{q_{1}} && \rnode{22}{\neg q_{1}}\\[.1cm] 
&\rnode{31}{p_{1}} &\\[.3cm] 
\rnode{41}{q_{2}} && \rnode{42}{\neg q_{2}}\\[.1cm] 
&\rnode{51}{\neg p_{1}} &
\ncline[doubleline=true,nodesep=5pt]{->}{11}{21}\mput*{A}
\ncline[doubleline=true,nodesep=5pt]{->}{11}{22}\mput*{B}
\ncline[doubleline=true,nodesep=5pt]{->}{21}{31}\mput*{C}
\ncline[doubleline=true,nodesep=5pt]{->}{22}{31}\mput*{D}
\ncline[doubleline=true,nodesep=5pt]{->}{31}{41}\mput*{E}
\ncline[doubleline=true,nodesep=5pt]{->}{31}{42}\mput*{F}
\ncline[doubleline=true,nodesep=5pt]{->}{41}{51}\mput*{G}
\ncline[doubleline=true,nodesep=5pt]{->}{42}{51}\mput*{H}
\end{array}
\]

In the case of $\Psi_{1}$, the labels in $A$, $B$, $C$ and $D$ are
$a_{1},\neg a_{1}$ and $E$, $F$, $G$ and $H$, are $a_{2},\neg a_{2}$.

In the case of $\Psi_{2}$, the labels in are shown below,
$A= a_{1},\neg a_{1}$, $B= a_{2}, \neg a_{2}$,
$C= a_{2}, \neg a_{2}$, $D= a_{1}, \neg a_{1}$
$E= a_{1}, \neg a_{1}, a_{2}, \neg a_{2}$, $F= a_{2}$
$G= a_{2}$ $H= a_{1}, \neg a_{1}, \neg a_{2}$.

After lifting and adding labels, we obtain,
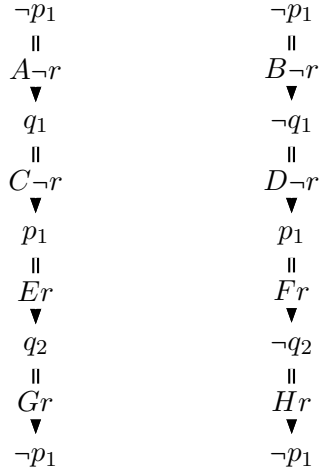
\[
\begin{array}{c@{\hskip 40pt}c@{\hskip 40pt}c}
\rnode{111}{\neg p_{1}} &&\rnode{112}{\neg p_{1}}\\[1cm] 
\rnode{21}{q_{1}} && \rnode{22}{\neg q_{1}}\\[1cm] 
\rnode{311}{p_{1}} &&\rnode{312}{p_{1}}\\[1cm] 
\rnode{41}{q_{2}} && \rnode{42}{\neg q_{2}}\\[1cm] 
\rnode{511}{\neg p_{1}} &&\rnode{512}{\neg p_{1}}
\ncline[doubleline=true,nodesep=5pt]{->}{111}{21}\mput*{A\neg r}
\ncline[doubleline=true,nodesep=5pt]{->}{112}{22}\mput*{B\neg r}
\ncline[doubleline=true,nodesep=5pt]{->}{21}{311}\mput*{C\neg r}
\ncline[doubleline=true,nodesep=5pt]{->}{22}{312}\mput*{D\neg r}
\ncline[doubleline=true,nodesep=5pt]{->}{311}{41}\mput*{E r}
\ncline[doubleline=true,nodesep=5pt]{->}{312}{42}\mput*{F r}
\ncline[doubleline=true,nodesep=5pt]{->}{41}{511}\mput*{G r}
\ncline[doubleline=true,nodesep=5pt]{->}{42}{512}\mput*{H r}
\end{array}
\]

Case $\Psi_{2}$, there are compatible combinations and, in the case of
$\Psi_{1}$, there is no compatible combination. We can infer $\Psi_{1}$ is not
satisfiable and, in the case of $\Psi_{2}$, the conclusion is based in writing
all combinations. 
\bibliographystyle{plain}
\addcontentsline{toc}{section}{Bibliography}
\bibliography{macrotableau}

\pagebreak
\begin{theindex}

  \item $Ch(e)$, 8
  \item $Ch\mathcal{P}\mapsto \clsddg$, 8
  \item $Down_{a}$, 15
  \item $Gr(B)$, 14
  \item $Gr(B\star R)=\langle E_{B\star R},   E_{B\star R},label\rangle$, 
		14
  \item $Gr(R)$, 14
  \item $Linclsd$, 17
  \item $Proj(e')$, 18
  \item $Proj:Linclsd\rightarrow\clsddg$, 18
  \item $Up_{a}$, 15
  \item $[\neg p_{l},p_{l}]\oplus   [p_{l},\neg p_{l}]$, 7
  \item $\Sigma$, set of all antichains in $\mathcal{P}$, 18
  \item $\Theta$, 18
  \item $\clndr=\langle V,E,label\rangle$, 5
  \item $\clsddg$, 7
  \item $\mathcal{NEC}$, 7
  \item $\mathcal{P}$, 8
  \item $maxG(v)$, 24
  \item \em branch, 13
  \item \thsat\ formula, 4
  \item \twsat\ formula, 4

  \indexspace

  \item addition of the new pair of literals to the set   of branching, 
		15
  \item antichain, 9

  \indexspace

  \item branching, 12

  \indexspace

  \item chain, 7
  \item chains, 8
  \item clause, 4
  \item Companion Digraph, 22
  \item compatible choice, 8
  \item compatible nested digraph of length $p$, 23
  \item compatible set of   entries, 8
  \item complete \thsat, 5
  \item cylindrical   digraph, 5

  \indexspace

  \item dept, 22

  \indexspace

  \item edge-labeled digraph for $v$, 22
  \item entails, 9
  \item entries, 4
  \item entry, 6
  \item Expansion of the Closed Digraph, 8
  \item expansion to $\mathcal{P}$, 8

  \indexspace

  \item incompatible, 8
  \item intersection of digraphs, $Gr_{v}\cap Gr'_{v}$, 23
  \item interval $[p,q]$, 7
  \item involves, 19

  \indexspace

  \item label of the   vertex $v$, 22
  \item Lifting by $B$ and $R$, 14
  \item Lifting of $Gr$ by $B$, 14
  \item Lifting of $Gr$ by $R$, 14
  \item linear interval, 13
  \item Linearized Closed Digraph, 17

  \indexspace

  \item maximal   nested path over a nested digraph, 27
  \item maximal branch, 13
  \item Multiplication of a   branch $br$, 15

  \indexspace

  \item necessarily true  literals, 7
  \item necessarily true literals, 7
  \item nested digraph of length $p$, 23

  \indexspace

  \item path, 6
  \item Pivot, 4
  \item Pivoted \thsat\ formula, 1, 4

  \indexspace

  \item root interval, 13
  \item roots, 12

  \indexspace

  \item satisfiable, 3
  \item set of closed digraph, 7
  \item set of entries of $\neg a_{i}$, 5
  \item set of entries of $a_{i}$, 5
  \item set of nested antichains, 27

  \indexspace

  \item tops, 12

  \indexspace

  \item union of digraphs, 23
  \item unsatisfiable, 3

\end{theindex}

\end{document}